\documentclass[sigconf]{acmart}
\AtBeginDocument{%
  \providecommand\BibTeX{{%
    \normalfont B\kern-0.5em{\scshape i\kern-0.25em b}\kern-0.8em\TeX}}}

\copyrightyear{2025}
\acmYear{2025}
\setcopyright{acmlicensed}\acmConference[KDD '25]{Proceedings of the 31st ACM SIGKDD Conference on Knowledge Discovery and Data Mining V.1}{August 3--7, 2025}{Toronto, ON, Canada}
\acmBooktitle{Proceedings of the 31st ACM SIGKDD Conference on Knowledge Discovery and Data Mining V.1 (KDD '25), August 3--7, 2025, Toronto, ON, Canada}
\acmDOI{10.1145/3690624.3709293}
\acmISBN{979-8-4007-1245-6/25/08}

\newtheorem{assumption}{Assumption}
\usepackage{algorithm}
\usepackage{algorithmicx}
\usepackage{algpseudocode}
\usepackage{float}
\usepackage{booktabs}
\usepackage{multirow}
\newtheorem{proposition}{Proposition}
\usepackage{colortbl}
\definecolor{mygray}{gray}{.9}
\usepackage{subfigure}
\usepackage{multicol}
\begin{document}

\title{Robust Uplift Modeling with Large-Scale Contexts for\\  Real-time Marketing}
% (KDD 2024 ID: 121)

\author{Zexu Sun}
\affiliation{%
  \institution{Gaoling School of Artificial Intelligence, Renmin University of China}
  \city{Beijing}
  \country{China}
}
\email{sunzexu21@ruc.edu.cn}

\author{Qiyu Han}
\affiliation{%
  \institution{School of Statistics, Renmin University of China}
  \city{Beijing}
  \country{China}
}
\email{hanqiyu@ruc.edu.cn}

\author{Minqin Zhu}
\affiliation{%
  \institution{Department of Computer Science and Technology, Zhejiang University}
  \city{Hangzhou}
  \country{China}
}
\email{minqinzhu@zju.edu.cn}

\author{Hao Gong}
\affiliation{%
  \institution{Kuaishou Technology}
  \city{Beijing}
  \country{China}
}
\email{gong_h@mail.nwpu.edu.cn}

\author{Dugang Liu}
\authornote{Corresponding Authors}
\affiliation{%
  \institution{College of Computer Science and Software Engineering, Shenzhen University}
  \city{Shenzhen}
  \country{China}}
\email{dugang.ldg@gmail.com}

\author{Chen Ma}
\authornotemark[1]
\affiliation{%
  \institution{Department of Computer Science, City University of Hong Kong}
  % \city{Hong Kong}
  \country{Hong Kong SAR}
}
\email{chenma@cityu.edu.hk}

\renewcommand{\shortauthors}{Zexu Sun, et al.}

\begin{abstract}
Improving user engagement and platform revenue is crucial for online marketing platforms. Uplift modeling is proposed to solve this problem, which applies different treatments (e.g., discounts, bonus) to satisfy corresponding users. Despite progress in this field, limitations persist. Firstly, most of them focus on scenarios where only user features exist. However, in real-world scenarios, there are rich contexts available in the online platform (e.g., short videos, news), and the uplift model needs to infer an incentive for each user on the specific item, which is called real-time marketing. Thus, only considering the user features will lead to biased prediction of the responses, which may cause the cumulative error for uplift prediction. Moreover, due to the large-scale contexts, directly concatenating the context features with the user features will cause a severe distribution shift in the treatment and control groups. Secondly, capturing the interaction relationship between the user features and context features can better predict the user response. To solve the above limitations, we propose a novel model-agnostic Robust Uplift Modeling with Large-Scale Contexts (UMLC) framework for Real-time Marketing. Our UMLC includes two customized modules. 1) A response-guided context grouping module for extracting context features information and condensing value space through clusters. 2) A feature interaction module for obtaining better uplift prediction. Specifically, this module contains two parts: a user-context interaction component for better modeling the response; a treatment-feature interaction component for discovering the treatment assignment sensitive feature of each instance to better predict the uplift. Moreover, we conduct extensive experiments on a synthetic dataset and a real-world product dataset to verify the effectiveness and compatibility of our UMLC.
%==
%Our code will be public after the paper is accepted.
\end{abstract}


% \ccsdesc[300]{Computational advertising}
% \ccsdesc{Applied computing}



\settopmatter{printfolios=true}
\maketitle

\section{Introduction}\label{sec:intro}
With the development of online platforms, assigning specific incentives to increase user engagement and platform revenue is a crucial task in online marketing~\cite{liu2023explicit,sun2023robustness}. 
%==
Since these incentives (\textit{e.g.}, bonus, discount) are usually related to the cost, different users may have different responses to these incentives. 
%==
Then, a successful strategy is needed to discover the incentive-sensitive user group and only deliver the incentive to the user that tends to be converted, which is important in online platforms.
%==
In recent years, uplift modeling, where various treatments are applied for different target users, has been proposed to achieve this purpose.  
%==
In practice, we can observe only one type of user response, which may be for a certain incentive (\textit{i.e.}, treatment group) or no incentive (\textit{i.e.}, control group)~\cite{diemert2018large}. 
%==
Uplift modeling aims to capture the differences between the treatment and control groups, which is known as the individual treatment effect (ITE) or uplift. 
%==
Many uplift models have been proposed to facilitate online marketing of online platforms~\cite{zhang2021unified}.

The existing works can be divided into two directions. 
1) \textit{Machine-learning based methods}. In detail, these works can be further divided into meta-learner-based and tree-based. 
% %==
% The basic idea of the meta-learner based methods is that using any off-the-shelf machine learning method as the base learner. 
%==
S-Learner~\cite{kunzel2019metalearners} and T-Learner~\cite{kunzel2019metalearners} are the representative meta-learner methods, which design a global base learner or two base learners for the samples in treatment and control groups, respectively. 
%==
Uplift Tree~\cite{rzepakowski2012decision} and Causal Forest~\cite{athey2019estimating} are two commonly used tree-based methods. 
% %==
% The basic idea of tree-based methods are to design the better split criterion or the ensemble formulation for more accurate uplift prediction.
%==
Specifically, the hierarchical tree is used to separate user populations into subgroups that exhibit sensitivity to specific treatments~\cite{radcliffe2011real}. 
%==
2) \textit{Representation-learning based methods}. With the development of deep learning, the representation learning becomes the main research direction for uplift prediction~\cite{mouloud2020adapting}. 
%==
% Some works of this direction aim to design the model structure to achieve the representation balancing in the latent space~\cite{zhang2021unified}. 
%==
TARNet~\cite{shalit2017estimating} and CFRNet~\cite{shalit2017estimating} leverage a feature representation network to extract the feature information into the latent space, then to predict the responses in different user groups (\textit{i.e.}, treatment group and control group) and achieve representation balancing by using the Integral Probability Metrics (IPM).
%==
StableCFR~\cite{wu2023stable} smooths the population distribution and upsample the underrepresented subpopulations, while balancing distributions between treatment and control groups.
%==
% There are also some other works focus on the other perspectives of uplift modeling, such as the generalization ability and the specific problems online marketing. 
% \cite{sun2023treatment} proposes a direct learning framework for handling the distribution shift in the test data. 
% %==
\cite{liu2023unite} designs a unified framework for both one-side and two-side online marketing.
% %==
We focus on the representation-learning based methods because they can be more flexibly adapted to modeling the complex scenarios in many industrial systems.

Although existing methods of uplift modeling have shown promising results and achieved success in many practical problems, they generally only consider the existence of user features and do not model the rich context features. 
%==
In some real-world scenarios, to better find the users that tend to convert in online platforms, we need to assign different treatments to users with different contexts.
%==
For example, on short video platforms, the users may receive different coin incentives with different recommended short videos.  
%==
We define this problem as real-time marketing. 
%==
There are some challenges for solving the uplift problem in real-time marketing:
%==
1) \textit{Distribution shift}: Random Control Trials (RCTs) are commonly used to collect data from the online platform for training the uplift models, which aim to ensure uniform distribution of user features between treatment and control groups. 
%==
However, as shown in Figure~\ref{fig:example}, when considering the large-scale context features, directly concatenating the uncontrollable context features (e.g., short videos) will cause a significant distribution shift between the treatment and control groups. 
%==
2) \textit{Feature interaction}: Modeling the feature interaction~\cite{wang2021dcn,sun2021fm2} can improve the prediction accuracy of the model. However, it is not sufficiently explored in previous uplift modeling works~\cite{belbahri2021qini,athey2019estimating}. With rich context features, discovering the relationships between the user features and context features, or the treatment assignment and the features, is essential for better predicting the uplift.

\begin{figure}[!t]
    \centering
    \includegraphics[width=\linewidth]{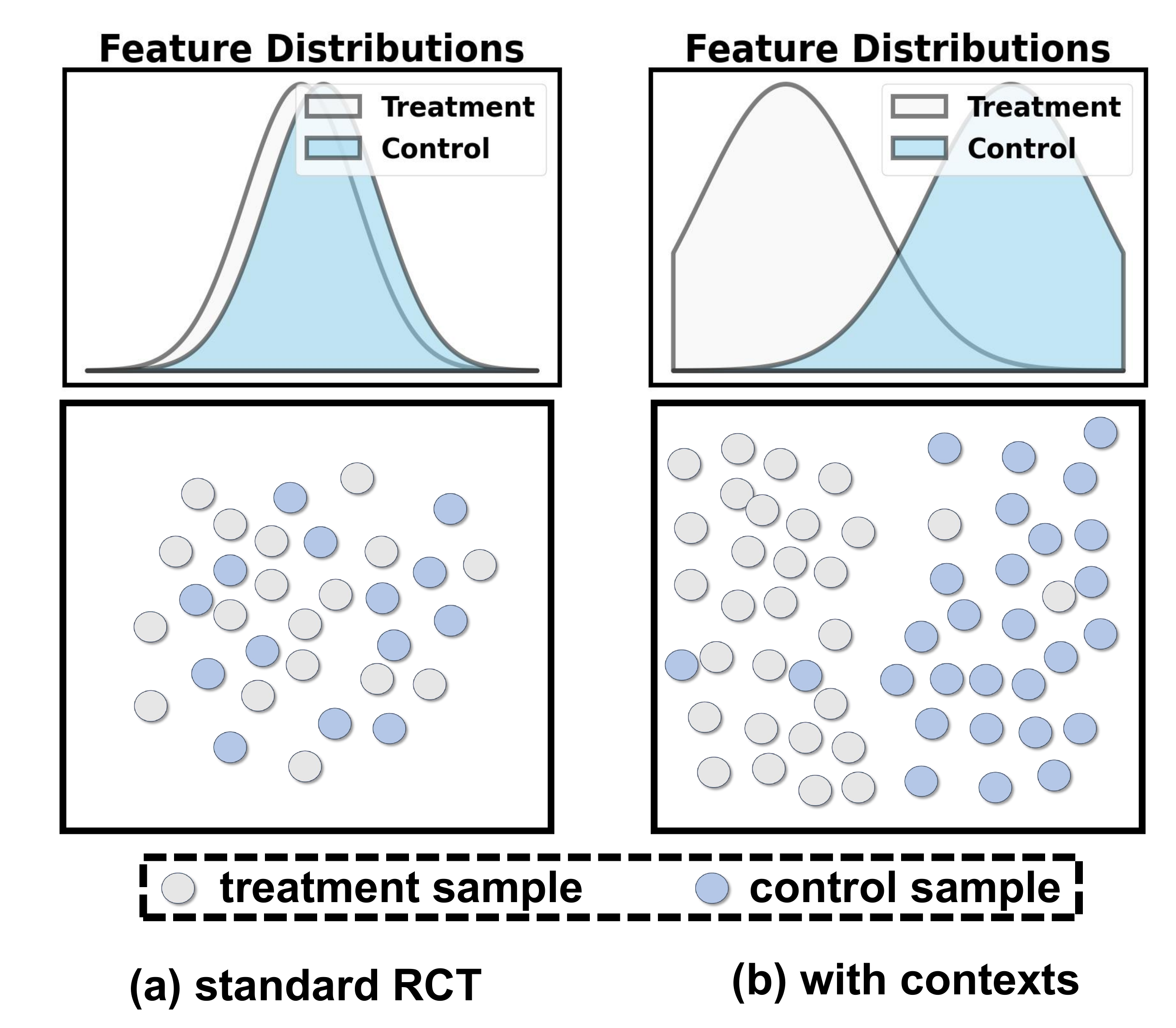}
    \vspace{-0.8cm}
    \caption{An example of distributions of standard RCTs and the RCTs considering the contexts. When considering the uncontrollable contexts in standard RCTs, there may be the significant distribution shift between the treatment and control groups.}
    \label{fig:example}
    \vspace{-0.8cm}
\end{figure}

In this paper, to solve the aforementioned challenges, we propose a novel model-agnostic Robust \underline{U}plift \underline{M}odeling with \underline{L}arge-Scale \underline{C}ontexts (UMLC) framework for Real-time Marketing.
%==
In particular, our framework consists of two customized modules. 
%==
1) A response-guided context grouping module aims to extract context features information through a response-guided regression network, and narrow the value space of context features by the cluster; 
%==
2) A feature interaction module contains a user-context and a treatment-feature interaction parts can better capture the relationship between the user features and context features, which may obtain more reliable user response prediction;
%==
a treatment-feature interaction module aims to discover the treatment assignment sensitive features (\textit{i.e.}, the features have a potential relationship with treatment assignment) of each instance to target the users that tend to convert with the specific context features.
Moreover, our contributions can be summarized as follows:
\begin{itemize}
    \item We propose a novel model-agnostic UMLC framework for solving the uplift problem in real-time marketing.
    \item We design a response-guided context grouping module, which handles the large-scale context with a distribution shift between the treatment and control groups. 
    \item We also design a feature interaction module, which builds the relationships between the user and context or treatment assignment and features for better uplift prediction.
    \item We conduct extensive experiments on a synthetic dataset and a product dataset to verify the effectiveness and compatibility of our UMLC. 
\end{itemize}

\section{Related Work}
% In this section, we briefly review the related works from two perspectives, which is the real-time marketing and uplift modeling.
\subsection{Real-time Marketing}
Real-time marketing appears as a result of consumers’ expectation for real-time connection
with product or service providers~\cite{buhalis2019real}. 
%==
Real-time service has taken personalization of service beyond merely providing relevant content to targeted consumers.
%==
For example, consumers’ personal preferences combined with external factors such as weather, traffic or season could inform context information needed to co-create value in real-time~\cite{buhalis2015socomo}.
%==
With the development of online marketing, more and more platforms need to design effective strategy for incentivize sensitive users, including E-commerce companies like Taobao, rental service like Airbnb, short video recommendation like Kuaishou. 
%==
The above mentioned platforms involve both user and rich contexts.
%==
To further better utilize the incentives, a possible marketing strategy is to assign different incentives for each user with different contexts, which is defined as the real-time online marketing.
%==
Recent works~\cite{jabbar2020real,liu2017distributed,liu2015bidding} in real-time marketing focus on the pricing for different contexts, the product management or the optimization procedure for bidding. 
%==
However, for the online platforms, the aim of real-time online marketing is to delivery different incentives for each user with different contexts, which is not investigated in the previous works.

\subsection{Uplift Modeling}
% In recent years, with the development of online marketing, uplift modeling is proposed to better assign treatment for each users, which can identify the corresponding sensitive population through estimating the ITE or uplift.
% %==
The uplift modeling methods mainly focus on two research directions: 1) \textit{Machine-learning based methods}. The machine-learning based methods can be divided into two more detailed research categories, \textit{i.e.}, meta-learner based and tree-based. 
%==
Meta-learner based methods aim to use any off-the-shelf estimator as the base learner for the uplift prediction. S-Learner~\cite{kunzel2019metalearners} and T-Learner~\cite{kunzel2019metalearners} leverage one global or two base learners for treatment and control groups, respectively. 
%==
% X-Learner~\cite{kunzel2019metalearners} is proposed to solve the selection bias between the two groups.
% %==
% However, when the uplift is relatively small, X-Learner performs inferior. 
% %==
% Then the R-Learner~\cite{nie2021quasi} is proposed to solve this problem, it utilizes Robinson Decomposition~\cite{robinson1988root} and adopts the idea of Double Machine Learning (DML)~\cite{chernozhukov2018double} in predicting uplift precisely.
% %==
Tree-based methods aim to design the splitting criteria or use the ensemble methods to predict the uplift.
%==
% Uplift tree~\cite{rzepakowski2012decision} designs the splitting criteria to maximize the gain in the distributional disparity between the treatment group's and the control group's responses.
%==
Causal Forest~\cite{athey2019estimating} uses the ensemble of uplift trees, and returns the average uplift prediction of them.
%==
% Due to the superiority of Causal Forest, there are some variants based on it. 
% %==
% Generalized Random Forest~\cite{athey2019generalized} extents the Random Forest~\cite{rigatti2017random}, which can estimate any interested metric, then Causal Forest becomes a special case of it.
%==
% Based on Generalized Random Forest, Orthogonal Random Forest~\cite{oprescu2019orthogonal} is a DML method, which adds the residuals into the score function, and combines the local weight for the uplift prediction. 
%==
2) \textit{Representation-learning based methods.} Representation-learning based methods are more flexible than machine-learning based methods. 
% Some works focus on achieving representation balancing through closing up the distributions between the treatment and control groups. 
%==
CFRNet~\cite{shalit2017estimating} use a shared bottom to extract features into latent space, and obtain the balanced representation through IPM.
%==
SITE~\cite{yao2018representation} preserves local similarity and balances data distributions simultaneously, by focusing on several hard samples in each mini-batch.
%==
StableCFR~\cite{wu2023stable} smooths the population distribution and upsample the underrepresented subpopulations.
%==
There are also some works focus on the specific scenario or the generalization ability of the uplift models.
% %==
UniTE~\cite{liu2023unite} proposes a uniform framework for one-side and two-side online marketing.
% %==
% \cite{sun2023treatment} proposes a domain adaptation framework for handling the distribution shift in the test data.
RUAD~\cite{sun2023robustness} defines a robustness case of uplift modeling and uses the adversarial training to solve it. 
%==
However, these works only model the uplift of the users, without considering the effect of contexts.

\section{Preliminaries}
To formulate the problem, we follow the Neyman-Rubin potential outcome framework~\cite{rubin2005causal}, to define the uplift modeling problem. 
% ==
Let the observed sample set be $\mathcal{D}=\{\boldsymbol{x}^u_i, \boldsymbol{x}^c_i, t_i, y_i\}^n_{i=1}$ with $n$ samples. 
% ==
Without loss of generality, for each sample, assuming $y_i\in \mathcal{Y}\subset \mathbb{R}$ is a continuous response variable, $\boldsymbol{x}^u_i \in \mathcal{X}^u\subset \mathbb{R}^p$ is a vector of user features with $p$ elements, $\boldsymbol{x}^c_i \in \mathcal{X}^c\subset \mathbb{R}^q$ is a vector of context features with $q$ elements, and $t_i\in \{0,1\}$ denotes the treatment indicator variable, \textit{i.e}., whether to get an incentive delivery.
% ==
Note that the proposed framework can be easily extended to other types of uplift modeling problems.
% ==
For user $i$, the change in user response caused by an incentive $t_i$, \textit{i.e.}, individual treatment effect or uplift, denoted as $\tau_i$, is defined as the difference between the treatment response and the control response:
\begin{equation}\label{eq:ite}
\tau_i=y_i(1)-y_i(0),
\end{equation}
where $y_i{(0)}$ and $y_i{(1)}$ are the user responses of the control and treatment groups, respectively. 

% ==
In the ideal world, \textit{i.e.}, obtaining the responses of a user in both groups simultaneously, we can easily determine the uplift $\tau_i$ based on Eq.~\eqref{eq:ite}. 
% ==
However, in the real world, only one of the two responses is observed for any one user. 
% ==
For example, if we have observed the response of a customer who receives the discount, it is impossible to observe the response of the same customer when they do not receive a discount, where such responses are often referred to as counterfactual responses. 
% ==
Therefore, the observed response can be described as:
\begin{equation}
y_i=t_i y_i{(1)}+(1-t_i) y_i{(0)}.
\end{equation}
For the brevity of notation, we will omit the subscript $i$ in the following if no ambiguity arises.

% ==
As mentioned above, the uplift $\tau$ is not identifiable since the observed response $y$ is only one of the two necessary terms (\textit{i.e}., $y(1)$ and $y(0)$).
% ==
Fortunately, with some appropriate assumptions~\cite{liu2023explicit}, we can use the conditional average treatment effect (CATE) as an estimator for the uplift, where CATE is defined as:
\begin{equation}
\begin{aligned}
\tau(\boldsymbol{x}) & =\mathbb{E}\left(Y{(1)} \mid \boldsymbol{X}=\boldsymbol{x}\right)-\mathbb{E}\left(Y{(0)} \mid \boldsymbol{X}=\boldsymbol{x}\right) \\
& =\underbrace{\mathbb{E}(Y \mid T=1, \boldsymbol{X}=\boldsymbol{x})}_{\mu_1(\boldsymbol{x})}-\underbrace{\mathbb{E}(Y \mid T=0, \boldsymbol{X}=\boldsymbol{x}) }_{\mu_0(\boldsymbol{x})}.
\end{aligned} \label{eq:uplift}
\end{equation}
where $\boldsymbol{x} = [\boldsymbol{x}^u,\boldsymbol{x}^c]$ is the features of an instance, $[\cdot,\cdot]$ represents the concatenation.
%==
Intuitively, the desired objective can be described as the difference between two conditional means $\tau(\boldsymbol{x})=\mu_1(\boldsymbol{x})-\mu_0(\boldsymbol{x})$.

\section{Methodology}
Our UMLC consists of two customized modules: the response-guided context grouping module and a feature interaction module.
%==
The whole structure of our UMLC is shown in Figure~\ref{fig:model}, and we introduce the two modules as follows.
\begin{figure*}[!t]
    \centering
    \includegraphics[width=0.9\linewidth]{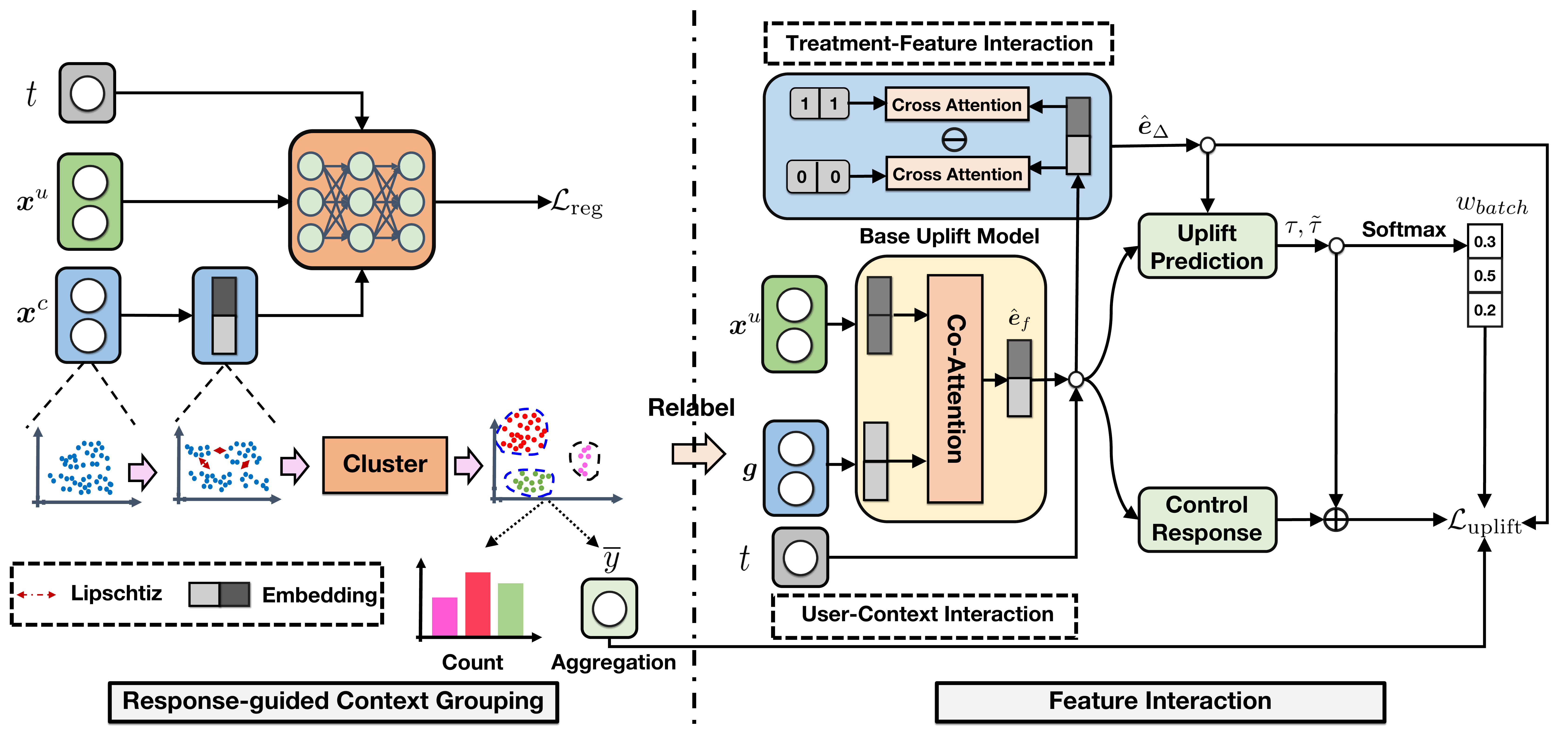}
        \vspace{-0.5cm}
    \caption{The overall structure of our UMLC framework. The left is the response-guided context grouping module, the right is the feature interaction module.}
    \label{fig:model}
    % \vspace{-0.5cm}
\end{figure*}
\subsection{Response-guided Context Grouping}\label{sec:group}

\subsubsection{Motivation} 
As mentioned in Section~\ref{sec:intro}, when considering the large-scale contexts, there is a \textit{distribution shift} between treatment and control groups, which may introduce bias for uplift prediction.
%==
In recent years, there have been works focusing on solving the treatment assignment bias between the user groups to obtain the unbiased treatment effect~\cite{shalit2017estimating,yao2018representation,wu2023stable}. 
%==
The core ideas of these works are to learn a balanced representation in the latent space from a distribution perspective~\cite{shalit2017estimating}, or match the samples in the treatment and control groups from a sample perspective~\cite{yao2018representation}.

However, in our problem setting, the above-mentioned methods that ignore the specific influence of the contexts may be unreliable. 
%==
For example, in real-world scenarios, each user can be associated with many contexts, such as a user may view many short videos a day, resulting in numerous instances. 
%==
Therefore, directly concatenating the context features with user features may encounter the issue of \textit{variance inflation}: similar context features may preserve many different labels, leading to unreliable predictions.

To address the issue, we propose to consolidate contexts into discrete groups, and then use these groups as proxy contexts to enhance uplift prediction. 
%==
Employing this strategy has the potential to reduce variance, but if the contexts in the same group are completely different, it may also introduce additional bias into the prediction of uplift~\cite{peng2023offline}.
%==
Fortunately, if the merged contexts within the same group have a similar impact on the response, the resulting bias can be effectively mitigated.

\subsubsection{Analysis}

For each context $\boldsymbol{x}^c \in \mathcal{X}^c$, we define a corresponding group variable $\boldsymbol{g} \in \{0,1, 2, \ldots, K-1\}$ denoting the group that the contexts are clustered into.
%==
Therefore, regarding the group variable $\boldsymbol{g}$ as a one-hot vector, we can define the conditional distribution on the group variable given context as $p\left(\boldsymbol{g} \mid \boldsymbol{x}^c\right)$. To support our method, we have the following assumptions.
% \begin{equation}
% p\left(\boldsymbol{g} \mid \boldsymbol{x}^c\right)= \begin{cases}1 & \text { if } \boldsymbol{g} \text { is the group of } x^c \\ 0 & \text { otherwise. }\end{cases}
% \end{equation}
% where $i \in \{1,2, \ldots, K\}$. 

% The joint distribution on user features, context features, treatment, group and response can also be defined as $p(\boldsymbol{x}^u, \boldsymbol{x}^c, t, \boldsymbol{g}, y)=p(\boldsymbol{x}^u, \boldsymbol{x}^c, t,  y) p(\boldsymbol{g} & \boldsymbol{x}^c)$, where $p(\boldsymbol{x}^u, \boldsymbol{x}^c, t, y)$ is the data distribution collected by the RCTs. 

%==
\begin{assumption}\label{assum1}
We assume the response generation is of the form $y=h(\boldsymbol{x}^u, \boldsymbol{x}^c, t)+\epsilon$~\cite{hoyer2008nonlinear}, where $|h(\boldsymbol{x}^u, \boldsymbol{x}^c, t)| \leq B_h$ and $\epsilon$ is a noise term with zero mean. We also assume the contexts merged into one group have a similar effect on the response. It means that for arbitrary context pair $(\boldsymbol{x}^c_i, \boldsymbol{x}^c_j)$ satisfying $\mathbb{P}\left(\boldsymbol{g} |\boldsymbol{x}^c_i\right)=\mathbb{P}\left(\boldsymbol{g} |\boldsymbol{x}^c_j\right)$, we have
\begin{equation}
  \left|\mathbb{E}\left[y| \boldsymbol{x}^u, \boldsymbol{x}^c_i,t\right]-\mathbb{E}\left[y| \boldsymbol{x}^u, \boldsymbol{x}^c_j,t\right]\right| \leq \delta, \forall \boldsymbol{x}^u, t  ,i\neq j,
\end{equation}
where $|\cdot|$ denotes absolute value operation, $B_h$ and $\delta$ are two constants, and $\delta$ is less than $B_h$.  
\end{assumption}

With the above assumption, we attempt to depict the distance of function $h$ between different contexts and merge the contexts with small distances into groups. 
%==
Given that the function space $h$ is of infinite dimension, it is impractical to compute the distance within the original function space. Then, we propose a method that involves transforming the resultant function $h$ into a finite-dimensional embedding, allowing for distance calculations within the projected space.
%==
To realize this strategy, we first give the following assumption.
\begin{assumption}\label{assm2}
There exists a transformation function $\xi$ which satisfies $\forall \boldsymbol{x}^u,t, i\neq j \left|h\left(\boldsymbol{x}^u, \boldsymbol{x}^c_i, t\right)-h\left(\boldsymbol{x}^u,\boldsymbol{x}^c_j, t\right)\right| \leq \zeta \cdot\left\|\xi\left(\boldsymbol{x}^c_i\right)-\xi\left(\boldsymbol{x}^c_j\right)\right\|_2$ $+\eta$, $\zeta,\eta$ are two constants, $\|\cdot\|_2$ denotes the Euclidean distances~\cite{danielsson1980euclidean}. 
\end{assumption}

If a transformation $\xi$ is found to meet Assumption~\ref{assm2} with a small value $\eta$, then it is possible to perform clustering on the context embedding $\xi(\boldsymbol{x}_{c})$ using the Euclidean distance. Through this process, the maximum distance between embeddings within a cluster can be effectively controlled to a small constant $\kappa$. 
%==
Regarding the clusters as the context groups, the distance of regression function $\delta \leq \zeta \cdot \kappa+\eta$ is also small.
%==
Then we have the following proposition to find such a transformation function.
%==
% Regarding the clusters as the context groups, the distance of outcome function $\delta \leq \zeta \cdot \gamma+\eta$ is also small.
%==
% To seek such transformation, we start from the following proposition:
\begin{proposition}\label{prop1}
If we can find a predictive function $f$ and transformation function $\xi$ with Lipschitz constraints on contexts such that $|h\left(\boldsymbol{x}^u, \boldsymbol{x}^c, t\right)-f\left(\boldsymbol{x}^u, \xi(\boldsymbol{x}^c), t\right)| \leq \mu, \forall \boldsymbol{x}^u,\boldsymbol{x}^c,t$ and $\left|f\left(\boldsymbol{x}^u, \xi(\boldsymbol{x}_i^c), t\right)\right.$ $\left.- f\left(\boldsymbol{x}^u, \xi(\boldsymbol{x}_j^c), t\right)\right| \leq c \cdot\left\|\xi(\boldsymbol{x}_i^c)-\xi(\boldsymbol{x}_j^c)\right\|_2, \forall \boldsymbol{x}^u,\boldsymbol{x}^c_i, \boldsymbol{x}^c_j,t,i\neq j$. Then the function $\xi$ satisfies the Assumption~\ref{assm2} with $\zeta=c, \eta=2 \mu$.    
\end{proposition}
\noindent
We provide the proof in Appendix~\ref{app:proof}.

Inspired by the Proposition~\ref{prop1}, to train the context embedding to be clustered, we aim to learn the transformation $\xi$ and a Lipschitz regularized regressor $f$ with small predictive error (\textit{i.e.}, small $\mu$). 

\subsubsection{Procedure}
We use the categorical and numerical embeddings~\cite{gorishniy2022embeddings} for the context features, and we construct the regression model as $f$ in Proposition~\ref{prop1} which predicts the response based on the user features, context features, and treatment with Lipschitz regularization.

\noindent
\textbf{Response Prediction.} We denote $\xi_\theta(\boldsymbol{x}^c)$ as the concatenation of the categorical and numerical embeddings~\cite{gorishniy2022embeddings} for context feature embeddings, $\theta$ is the learnable parameters. 
%==
To ensure that the acquired embeddings $\xi_\theta\left(x^c\right)$ accurately capture the context effects on the responses and the distance between them, we concatenate the representation $\xi_\theta\left(x^c\right)$ with the user features and the treatment. This concatenation is then fed into the regression model $f$ to make predictions of the response.
%==
The prediction loss for the response-guided regression is defined as:
\begin{equation}
%\mathbb{E}_{(\boldsymbol{x}^u, \boldsymbol{x}^c, t, y)\sim \mathcal{D}}
\mathcal{L}_{\text{pred}}=\mathcal{L} \left(f\left(\boldsymbol{x}^u, \xi_{\theta}(\boldsymbol{x}^c), t\right),y\right),  
\label{eq:pred_loss}
\end{equation}
where $\mathcal{L}(\cdot, \cdot)$ is the mean square error loss. 

%==
\noindent
\textbf{Lipschitz Regularization.} 
To ensure a consistent representation of the relationship of the contexts in the embedding space, we employ Lipschitz regularization. This regularization guarantees that the function $f$ is Lipschitz continuous, meaning that the distance in the embedding space accurately reflects the context effects. To satisfy Lipschitz continuity, there must exist a non-negative constant $c$ for  $\forall i, j, i\neq j$ such that:
\begin{equation}
\left|f\left(\boldsymbol{z}_i\right)-f\left(\boldsymbol{z}_j\right)\right| \leq c\left\|\boldsymbol{z}_i-\boldsymbol{z}_j\right\|_2,   
\end{equation}
where $\boldsymbol{z}_i=\left(\boldsymbol{x}_i^u, \xi(\boldsymbol{x}_i^c), t_i\right)$ is the concatenation of the user features, the context embedding and the treatment, and intuitively, the change of response is bounded by constant $c$ for smoothness. 
%==
It is evident that if the function $f$ is $c$-Lipschitz on the input $\boldsymbol{z}$, then it is also $c$-Lipschitz on the context embedding $\xi(\boldsymbol{x}^c)$.

%==
We utilize Lipschitz Regularization~\cite{liu2022learning}, a technique that relies solely on the weight matrices of each network layer, with an estimated per-layer Lipschitz upper bound denoted as $c_i$. The purpose of this regularization is to constrain the discrepancies between context effects and responses. The corresponding loss function can be defined as follows:
\begin{equation}
\mathcal{L}_{\text {Lip }}=\prod_{i=1}^l \operatorname{softplus}\left(c_i\right),  
\label{eq:lips}
\end{equation}
where $\operatorname{softplus}\left(c_i\right)=\ln \left(1+e^{c_i}\right)$ is a reparameterization strategy to avoid invalid negative estimation on Lipschitz constant, $c_i$ and $l$ is the number of network layers.

Then we can obtain the final loss function for training the response-guided context embedding:
\begin{equation}
 \mathcal{L}_{\text{reg}}=\mathcal{L}_{\text{pred}}+\alpha \mathcal{L}_{\text {lip }}. \label{eq:reg_loss}
\end{equation}
where $\alpha$ is the hyper-parameter to control the trade-off of aligning the difference between the context embedding and the response. 
%==
Following the previous work~\cite{liu2022learning}, the value of $\alpha$ is stated to be minimal, thus, we set it to be $10^{-4}$ in all experiments. 
%==
% By simultaneously optimizing $\mathcal{L}_{\text{reg}}$ and $\mathcal{L}_{\text {lip }}$, we can obtain a predictive model and the embedding function satisfying the requirements in Proposition~\ref{prop1}.

\noindent
\textbf{Grouping and Aggregation.}
After we get the trained context embedding, grouping it is important for the data dimensionality reduction and volatility minimization. 
%==
We determine the number of groups by setting a hyperparameter $K$, and utilize clustering algorithms like K-means to allocate each context embedding $\xi_\theta(\boldsymbol{x}^c)$ to a specific group $\boldsymbol{g}$.
%==
The clustering mapping $\mathcal{F}$, which assigns embeddings to groups, is formalized as: $\mathcal{F}: \xi_\theta(\boldsymbol{x}^c) \mapsto \boldsymbol{g}$. 
% \begin{equation}
% h(\xi): \xi_\theta(\boldsymbol{x}^c) \mapsto \mathbf{G}, \text{ where } \xi \in \mathbb{R}^{\mathcal{D}_\xi} \text{ and } \mathbf{G} \text{ is the set of groups}
% \end{equation}

Moreover, to enhance the stability and accuracy of uplift model training and prediction
aggregating data with the same features is necessary.
%==
Specially, for any two samples $(\boldsymbol{x}^u_i, \boldsymbol{g}_i, t_i, y_i)$ and $(\boldsymbol{x}^u_j, \boldsymbol{g}_j,$ $t_j, y_j)$, $i\neq j$, if they have the same treatment, user features and context group, we aggregate them into a new sample as:
\begin{equation}
    \overline{y} = \frac{y_i + y_j}{2},\quad \forall \boldsymbol{x}^u_i = \boldsymbol{x}^u_j, \boldsymbol{g}_i=\boldsymbol{g}_j, t_i= t_j. \label{eq:group_data}
\end{equation}
Then, we can get the relabeled dataset $\mathcal{D}_r=\{\boldsymbol{x}^u_i, \boldsymbol{g}_i, t_i, \overline{y}\}_{i=1}^m$ with $m$ samples for subsequent uplift prediction.

\subsection{Feature Interaction}
The second problem we need to solve is the feature interaction in uplift prediction.
%==
Feature interaction is designed to model combinations between different features and have been shown to significantly improve the performance of a response model~\cite{calder2003feature}.
%==
To obtain better uplift prediction, we mainly focus on the user-context interaction and treatment-feature interaction, we detailed introduce them in the following.

\subsubsection{User-Context Interaction}

In real time marketing, when considering the context (\textit{i.e.}, short videos), the context information can influence the user behavior. 
%==
Recent state-of-the-art methods~\cite{liu2023explicit,gan2023making} represent the relationships between users and contexts using a tensor. However, utilizing a tensor makes it challenging to differentiate the effects of various context factors and to model the non-linear interactions between users and contexts.
%==
In this part, we use the user feature embedding $\boldsymbol{e}_u$ and the grouped context embedding $\boldsymbol{e}_c$ to model the user response.
%==
We adapt a parallel co-attention network~\cite{lu2016hierarchical} for the user-context interaction to get better prediction results, the first step is to compute the affinity matrix:
\begin{equation}
\boldsymbol{L} = \mathrm{tanh}\left(\boldsymbol{e}_u^{T} \cdot \boldsymbol{W}_L \cdot\boldsymbol{e}_c\right), 
\end{equation}
where $\boldsymbol{W}_L$ is the weight parameters. After computing this affinity matrix, 
%one possible way of computing the user (or context) attention is to simply maximize out the affinity over the locations of other modality, \textit{i.e.}, $\boldsymbol{a}_u[n]=\max _i\left(\boldsymbol{L}_{i, n}\right)$ and $\boldsymbol{a}_c[t]=\max_j\left(\boldsymbol{L}_{t, j}\right)$. 
%==
the parallel co-attention network considers this affinity matrix as a feature and learns to predict user and context attention maps via the following:
\begin{equation}
\begin{aligned}
  \boldsymbol{H}_{u}=&\mathrm{tanh}\left(\boldsymbol{W}_u \cdot \boldsymbol{e}_u + \left(\boldsymbol{W}_c \cdot \boldsymbol{e}_c \right)\cdot \boldsymbol{L}\right),\\
  \boldsymbol{H}_{c}=& \mathrm{tanh}\left(\boldsymbol{W}_c \cdot \boldsymbol{e}_c + \left(\boldsymbol{W}_u \cdot \boldsymbol{e}_u \right) \cdot \boldsymbol{L}^T\right).
\end{aligned}\label{eq:att_matrix}
\end{equation}
where $\boldsymbol{W}_u$ and $\boldsymbol{W}_c$ are the weight parameters. The affinity matrix $\boldsymbol{L}$ transform the context attention space to the user attention space (vice versa for $\boldsymbol{L}^T$). Next, we can get the normalized attention weights for the user and context embeddings:
\begin{equation}
    \boldsymbol{a}_u = \mathrm{softmax}\left(\boldsymbol{W}^T_{hu}\cdot \boldsymbol{H}_{u}\right),~~
    \boldsymbol{a}_c = \mathrm{softmax}\left(\boldsymbol{W}^T_{hc}\cdot \boldsymbol{H}_{c}\right),
\end{equation}
where $\boldsymbol{W}_{hu}$ and $\boldsymbol{W}_{hc}$ are the weight parameters, $\boldsymbol{a}_u$ and $\boldsymbol{a}_c$ are the attention parameters for the user and context, respectively.
%==
Based on the above attention weights, the user and context attention vectors are calculated in the following:
\begin{equation}
    \hat{\boldsymbol{e}}_u = \boldsymbol{a}_u * \boldsymbol{e}_u,
    \quad
    \hat{\boldsymbol{e}}_c = \boldsymbol{a}_c * \boldsymbol{e}_c, \label{eq:atten_vec}
\end{equation}
where $*$ denotes the sum of the product for each element in the embedding with the attention vector. Then, we use a multilayer perceptron (MLP) to predict control responses:
\begin{equation}
\mu_0=\mathrm{MLP}(\hat{\boldsymbol{e}}_u, \hat{\boldsymbol{e}}_c).\label{eq:pred_control}
\end{equation}
\subsubsection{Treatment-Feature Interaction}
% In this module, we aim to use the user embedding $\hat{\boldsymbol{e}}_u$ and context embedding $\hat{\boldsymbol{e}}_c$ to predict the uplifts of users. 
%==
Considering the treatment assignment sensitive features in the concatenation of the user features and context features, we introduce a cross attention based treatment-feature interaction part to discover the information gain~\cite{tao2023event} introduced by the treatment assignment.

Specially, we denote the concatenation of $\hat{\boldsymbol{e}}_u$ and $\hat{\boldsymbol{e}}_c$ as $\hat{\boldsymbol{e}}_f$, the embedding of treatment as $\boldsymbol{e}_t$. The we can get the attention weights of the feature embedding as follows:
\begin{equation}
    \boldsymbol{a}_t = \mathrm{softmax}\left( \frac{\left(\boldsymbol{W}_t\cdot \boldsymbol{e}_t\right)\cdot \left(\boldsymbol{W}_f\cdot \hat{\boldsymbol{e}}_f\right)^T}{\sqrt{K_d}} \right),
\end{equation}
where $\boldsymbol{W}_t$ and $\boldsymbol{W}_f$ are the weight parameters, $K_d$ is the dimension of the output embedding.
%==
Then, to simulate the treatment variation, we can obtain the $\boldsymbol{a}_t^0$ and $\boldsymbol{a}_t^1$ with different treatment embeddings $\boldsymbol{e}^0_t$ (\textit{i.e.}, $t=0$)  and $\boldsymbol{e}^1_t$ (\textit{i.e.}, $t=1$), respectively. The information gain on the feature embedding introduced by the treatment assignment can be formulated as:
\begin{equation}
\hat{\boldsymbol{e}}_{\Delta} = \boldsymbol{a}_t^1 * \hat{\boldsymbol{e}}_f- \boldsymbol{a}_t^0 * \hat{\boldsymbol{e}}_f,
\label{eq:infor_gain}\end{equation}
to better optimize the cross attention structure to find treatment the sensitive features, we aim to maximize the information gain $\hat{\boldsymbol{e}}_{\Delta}$ in the final loss function.

Similar to the control response prediction, we can also use a multilayer perceptron to predict the uplift of each sample with and without the information gain, respectively, as defined below:
\begin{equation}
\hat{\tau} = MLP(\hat{\boldsymbol{e}}_{f}), \quad \tilde{\tau} = MLP(\hat{\boldsymbol{e}}_{\Delta} + \hat{\boldsymbol{e}}_{f}). \label{eq:uplift_pred}
\end{equation}

Thus, we can get the information gain coefficient through
a difference function, which acts as the importance of the
sample in the uplift prediction loss function, and is defined as follows:
\begin{equation}
\begin{aligned}
 w_{b a t c h}&=\frac{\exp \left(\tilde{\tau}-\hat{\tau}\right)}{\sum_{b a t c h} \exp \left(\tilde{\tau}-\hat{\tau}\right)}, \\
 %\frac{1}{M}\sum_{ (\boldsymbol{x}^u, \boldsymbol{g}, t, \overline{y})\sim \mathcal{D}_r}
\mathcal{L}_{\text {uplift }}&= w_{b a t c h} \cdot \left(\left(1-t\right)\cdot \mathcal{L}(\mu_0, \overline{y})\right.\\&\left.\qquad + t \cdot\left(\mathcal{L}(\mu_1, \overline{y}) +\beta\cdot \mathcal{L}(\tilde{\mu}_1, \overline{y}) \right) \right)- \gamma\cdot \|\hat{\boldsymbol{e}}_{\Delta}\|_F^2.
\end{aligned}\label{eq:uplift_loss}
\end{equation}
where $\mu_1 =\mu_0 + \hat{\tau}$ is the predicted treatment response, $\tilde{\mu}_1 =\mu_0 + \tilde{\tau}$ is the predicted treatment response with treatment-feature interaction, $w_{b a t c h}$ is the sample weight calculated in the batch, $\beta$ and $\gamma$ are the hyperparmeters to control the trade-off.

% \begin{algorithm}[h]
%     \centering
%     \caption{}
%     \begin{algorithmic}
%     \Require 111
%     \end{algorithmic}
% \end{algorithm}
\section{Experiments}
In this section, we conduct extensive experiments on a synthetic dataset and a product dataset to answer the following research questions:
\begin{itemize}
    \item \textbf{RQ1}: Can our UMLC outperform different baselines on various commonly used uplift modeling metrics?
    \item \textbf{RQ2}: How does each proposed module contributes to the performance of our UMLC?
    \item \textbf{RQ3}: How does the choice of group number influences the performance of our UMLC?
\end{itemize}
And we also conduct more detailed experiments and analysis in Appendix~\ref{app:exp}.

\subsection{Experimental Setups}

\subsubsection{Datasets}
\textbf{Synthetic dataset.}
To simulate complex real-world scenarios, the synthetic dataset used in our incorporates several key factors. 
%==
1) the generation of user features and context features is independent of the treatment to simulate the real-world RCTs. 
%==
2) For each user, a random selection of 60-130 contexts is made from the context pool, and these are combined to generate the samples, which ensures that each user is connected with multiple contexts in a random manner. 
%==
3) The distributions of user responses in both the treatment and control groups exhibit long-tailed behavior. In light of Assumption~\ref{assum1}, we divide the contexts into six groups based on their influence on the response. 
%==
4) Similar to the real-world data, the features include three types of distribution features: binary, categorical, and numerical. For each type, we use Bernoulli, Multinomial, and Gaussian distributions, respectively.
%==
Moreover, we present the detailed generation process of the Synthetic dataset in Appendix~\ref{app:syn}.
% In this study, we construct a synthetic dataset. At the heart of our methodology is the creation of a feature matrix X that includes a diverse range of features, integrating user characteristics with a randomly selected set of context attributes. This design effectively emulates a "one-to-many" relationship, capturing the nuanced interactions between a user and multiple contextual elements. 
% The outcome variables $Y_0$ and $Y_1$ are derived from the feature matrix X, accounting for the interactions between various types of user and context features, as well as the effects of an intervention. Our analysis employs a treatment variable T, which is randomly assigned to replicate the setup of experimental interventions. This structured approach enables us to comprehensively assess how the interplay of diverse features and treatments impacts the outcomes, paralleling the analytical depth of a randomized control trial.

\noindent
\textbf{Production dataset.} This dataset comes from an industrial production environment, one of the largest short-video platforms in China. For such short video platforms, clarity is an important user experience indicator. A decrease in clarity may lead to a decrease in users' playback time. Therefore, through random experiments within a week, we provided high-clarity videos ($t=1$) to the treatment group and low-clarity videos ($t=0$) to the control group. We count the total viewing time of users' short videos in a week and quantify the impact of definition degradation on user experience. 
%==
The visualization of this dataset is shown in Figure~\ref{fig:kuai}. 
%==
For the online random experiments, we control the distributions of the users in two groups to be similar and collect the data that each user with all the short videos played. 
%==
Moreover, the statistics of the two datasets are presented in Appendix~\ref{app:syn}.

\begin{figure}[!t]
    \centering
    \includegraphics[width=0.9\linewidth]{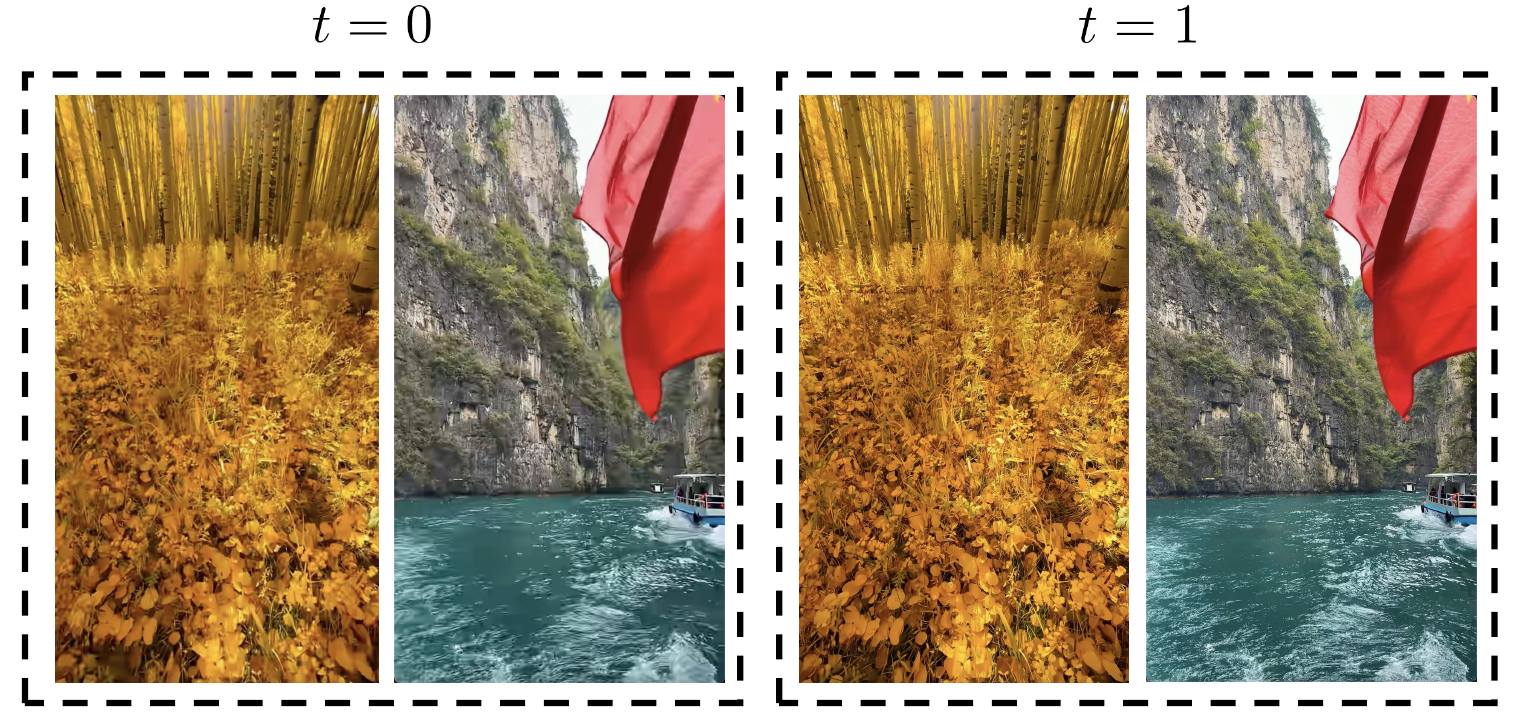}
    \caption{The visualization of the Production dataset. As $t$ increases, the clarity of the video correspondingly enhances.}
    \label{fig:kuai}
    \vspace{-0.6cm}
\end{figure}
\begin{table*}[!t]
    \centering
    \caption{Overall comparison between our UMLC and the baselines on Synthetic and Production datasets. We report the results over five random seeds. The best results and second best results are \textbf{bold} and \underline{underlined}, respectively.}\label{tab:overall}
    \begin{tabular}{c|ccccccc}\toprule
    \multirow{2}{*}{\textbf{Method}} & \multicolumn{3}{c}{Synthetic Dataset} & &\multicolumn{3}{c}{Production Dataset}  \\\cline{2-4} \cline{6-8}
    &AUUC & QINI & KENDALL & & AUUC & QINI & KENDALL \\\midrule
    S-Learner & 0.2104 $\pm$ 0.0292& 0.2502 $\pm$ 0.0116& 0.1186 $\pm$ 0.0118& & 1.7907 $\pm$ 0.0529 &2.3156 $\pm$ 0.0493 & 0.2631 $\pm$ 0.0252 \\
    T-Learner & 0.2250 $\pm$ 0.0201& \underline{0.2697} $\pm$ 0.0178& 0.1185 $\pm$ 0.0110& & 1.8007 $\pm$ 0.0634 &2.3426 $\pm$ 0.0504 & 0.3152 $\pm$ 0.0209 \\
    TARNet & 0.2131 $\pm$ 0.0136& 0.2379 $\pm$ 0.0224& 0.1368 $\pm$ 0.0122& & 1.7754 $\pm$ 0.0575 &2.2764 $\pm$ 0.0473 & 0.3263 $\pm$ 0.0286 \\
    CFRNet-mmd & 0.2077 $\pm$ 0.0237& 0.2120 $\pm$ 0.0257& 0.1689  $\pm$ 0.0134& & 1.6569 $\pm$ 0.0594 &2.1614 $\pm$ 0.0528 & 0.3368 $\pm$ 0.0227 \\
    CFRNet-wass & 0.2139 $\pm$ 0.0298& 0.2234 $\pm$ 0.0309& 0.1594 $\pm$ 0.0166& & 1.7286 $\pm$ 0.0646 &2.1093 $\pm$ 0.0602 & 0.3536 $\pm$ 0.0346 \\

    DragonNet & 0.2574 $\pm$ 0.0365& 0.2153 $\pm$ 0.0279& \underline{0.1754} $\pm$ 0.0138& & 1.3581 $\pm$ 0.0613 &1.8750 $\pm$ 0.0522 & \textbf{0.3894} $\pm$ 0.0258 \\
    EUEN & 0.2158 $\pm$ 0.0202& 0.2170 $\pm$ 0.0195& 0.1126 $\pm$ 0.0087& & 1.2506 $\pm$ 0.0505 &1.7092 $\pm$ 0.0513 & 0.3157 $\pm$ 0.0220 \\
    UniTE & 0.2091 $\pm$ 0.0226& 0.2071 $\pm$ 0.0234& 0.1284 $\pm$ 0.0126& & 1.3092 $\pm$ 0.0499 &1.8075 $\pm$ 0.0534 & \underline{0.3684} $\pm$ 0.0275 \\\midrule
 \rowcolor{mygray}    \textbf{UMLC} (CFRNet-mmd) & \textbf{0.3149} $\pm$ 0.0209& 0.2558 $\pm$ 0.0191& 0.1647 $\pm$ 0.0152& & \textbf{2.2106} $\pm$ 0.0595 & \textbf{2.6105} $\pm$ 0.0545 & 0.3473 $\pm$ 0.0246 \\
 \rowcolor{mygray}    \textbf{UMLC} (CFRNet-wass) & \underline{0.2871} $\pm$ 0.0235& 0.2499 $\pm$ 0.0204& 0.1584 $\pm$ 0.0112& & 2.0329 $\pm$ 0.0565 & \underline{2.5201} $\pm$ 0.0453 & 0.3210 $\pm$ 0.0277 \\
 \rowcolor{mygray}   \textbf{UMLC} (DragonNet) & 0.2549 $\pm$ 0.0208& \textbf{0.2961} $\pm$ 0.0216& \textbf{0.1894} $\pm$ 0.0165& & 1.9019 $\pm$ 0.0556 &1.8514 $\pm$ 0.0463 & 0.3368 $\pm$ 0.0265 \\
 \rowcolor{mygray}   \textbf{UMLC} (EUEN) & 0.2406 $\pm$ 0.0245& 0.2594 $\pm$ 0.0268&  0.1473 $\pm$ 0.0124& & 1.9539 $\pm$ 0.0555 &2.3680 $\pm$ 0.0469 & 0.2631 $\pm$ 0.0251 \\
 \rowcolor{mygray}      \textbf{UMLC} (UniTE) & 0.2533 $\pm$ 0.0269& 0.2467 $\pm$ 0.0221& 0.1594 $\pm$ 0.0166& & \underline{2.0505} $\pm$ 0.0594 &2.4416 $\pm$ 0.0569 & 0.2947 $\pm$ 0.0271 \\\midrule
    \end{tabular}
    % \vspace{-0.3cm}
\end{table*}
\begin{table*}[!t]
    \centering
    \caption{Ablation study of our UMLC with four base uplift models on Synthetic and Production datasets. We report the results over five random seeds. The best results and second best results are \textbf{bold} and \underline{underlined}, respectively.}    \label{tab:ablation}
    \resizebox{\linewidth}{!}{
    \begin{tabular}{c|ccccccc}\toprule
    \multirow{2}{*}{\textbf{Method}} & \multicolumn{3}{c}{Synthetic Dataset} & &\multicolumn{3}{c}{Production Dataset}  \\\cline{2-4} \cline{6-8}
    &AUUC & QINI & KENDALL & & AUUC & QINI & KENDALL \\\midrule
 \rowcolor{mygray}         \textbf{UMLC} (CFRNet-mmd) & \textbf{0.3149} $\pm$ 0.0209& \textbf{0.2558} $\pm$ 0.0191& \underline{0.1647} $\pm$ 0.0152& & \textbf{2.2106} $\pm$ 0.0595 & \textbf{2.6105} $\pm$ 0.0545 & \underline{0.3473} $\pm$ 0.0246   \\
    \textsc{w/o RCG} & 0.2321 $\pm$ 0.0201& 0.2230 $\pm$ 0.0208&  \textbf{0.1694} $\pm$ 0.0162& & 1.9073 $\pm$ 0.0549 &2.4665 $\pm$ 0.0435 & \textbf{0.3578} $\pm$ 0.0249 \\
    \textsc{w/o UCI} & \underline{0.2395} $\pm$ 0.0221& 0.2268 $\pm$ 0.0164& 0.1547 $\pm$ 0.0149& & 1.8578 $\pm$ 0.0584 &1.8255 $\pm$ 0.0470 & 0.3084 $\pm$ 0.0251 \\
   \textsc{w/o TFI} & 0.2144 $\pm$ 0.0213& \underline{0.2482} $\pm$ 0.0144& 0.1315 $\pm$ 0.0159& & \underline{2.1349} $\pm$ 0.0557 & \underline{2.4983} $\pm$ 0.0539 & 0.3263 $\pm$ 0.0248 \\
   \midrule
\rowcolor{mygray}         \textbf{UMLC} (CFRNet-wass) & \textbf{0.2871} $\pm$ 0.0235& \textbf{0.2499} $\pm$ 0.0204& \underline{0.1584} $\pm$ 0.0112& & \underline{2.0329} $\pm$ 0.0565 &\textbf{2.5201} $\pm$ 0.0453 & \underline{0.3210} $\pm$ 0.0277 \\
    \textsc{w/o RCG} & 0.2221 $\pm$ 0.0209& 0.2279 $\pm$ 0.0196& \textbf{0.1601} $\pm$ 0.0123& & 1.9214 $\pm$ 0.0568 &2.3612 $\pm$ 0.0519 & \textbf{0.3597} $\pm$ 0.0244 \\
    \textsc{w/o UCI} & \underline{0.2355} $\pm$ 0.0213& 0.2026 $\pm$ 0.0219& 0.1328 $\pm$ 0.0179& & 1.9307 $\pm$ 0.0583 &1.9992 $\pm$ 0.0478 & 0.3122 $\pm$ 0.0257 \\
   \textsc{w/o TFI} & 0.2062 $\pm$ 0.0204& \underline{0.2368} $\pm$ 0.0205& 0.1473 $\pm$ 0.0119& & \textbf{2.0519} $\pm$ 0.0549 & \underline{2.4805} $\pm$ 0.0543 & 0.3137 $\pm$ 0.0248 \\
   \midrule
\rowcolor{mygray}         \textbf{UMLC} (DragonNet) & 0.2549 $\pm$ 0.0208& \textbf{0.2961}  $\pm$ 0.0216& \textbf{0.1894}  $\pm$ 0.0165& & \underline{1.9019} $\pm$ 0.0556 & \textbf{1.8514} $\pm$ 0.0463 & 0.3368 $\pm$ 0.0265\\
    \textsc{w/o RCG} &  \textbf{0.3094} $\pm$ 0.0239& 0.2019 $\pm$ 0.0222& 0.1484 $\pm$ 0.0117& & 1.7542 $\pm$ 0.0574 & \underline{1.8481} $\pm$ 0.0552 & \textbf{0.4210} $\pm$ 0.0252 \\
    \textsc{w/o UCI} & 0.2634 $\pm$ 0.0212& 0.2064 $\pm$ 0.0258& 0.1557 $\pm$ 0.0174& & 1.8401 $\pm$ 0.0572 &1.8370 $\pm$ 0.0487 & \underline{0.3398} $\pm$ 0.0251 \\
   \textsc{w/o TFI} & \underline{0.2639} $\pm$ 0.0244& \underline{0.2365} $\pm$ 0.0205& \underline{0.1778} $\pm$ 0.0149& & \textbf{1.9050} $\pm$ 0.0583 &1.8405 $\pm$ 0.0538 & 0.2778 $\pm$ 0.0277 \\
   \midrule
\rowcolor{mygray}         \textbf{UMLC} (EUEN) &  \textbf{0.2406} $\pm$ 0.0245& \textbf{0.2594} $\pm$ 0.0268& \underline{0.1473} $\pm$ 0.0124& & \textbf{1.9539} $\pm$ 0.0555 & \textbf{2.3680} $\pm$ 0.0469 & 0.2631 $\pm$ 0.0251 \\
    \textsc{w/o RCG} & \underline{0.2295} $\pm$ 0.0251& \underline{0.2404} $\pm$ 0.0226& 0.1263 $\pm$ 0.0121& & \underline{1.5694} $\pm$ 0.0586 &2.0819 $\pm$ 0.0555 & \underline{0.3789} $\pm$ 0.0257 \\
    \textsc{w/o UCI} & 0.2170 $\pm$ 0.0215& 0.1783 $\pm$ 0.0251& \textbf{0.1491} $\pm$ 0.0143& & 1.3145 $\pm$ 0.0568 &1.8223 $\pm$ 0.0538 &  \textbf{0.3992} $\pm$ 0.0265 \\
   \textsc{w/o TFI} & 0.2119 $\pm$ 0.0235& 0.2121 $\pm$ 0.0223& 0.1157 $\pm$ 0.0135& & 1.4179 $\pm$ 0.0579 & \underline{2.2572} $\pm$ 0.0492 & 0.2315 $\pm$ 0.0266 \\
   \midrule   
\rowcolor{mygray}         \textbf{UMLC} (UniTE) & \textbf{0.2533} $\pm$ 0.0269& \textbf{0.2467} $\pm$ 0.0221& \underline{0.1594} $\pm$ 0.0166& & \textbf{2.0505} $\pm$ 0.0594 & \textbf{2.4416} $\pm$ 0.0569 & 0.2947 $\pm$ 0.0271  \\
    \textsc{w/o RCG} & 0.2118 $\pm$ 0.0217& 0.2083 $\pm$ 0.0244& 0.1401 $\pm$ 0.0135& & 1.6776 $\pm$ 0.0576 &2.1548 $\pm$ 0.0526 & \underline{0.3421} $\pm$ 0.0277 \\
    \textsc{w/o UCI} & 0.2087 $\pm$ 0.0295& \underline{0.2204} $\pm$ 0.0226& 0.1389 $\pm$ 0.0164& & 1.7011 $\pm$ 0.0574 & \underline{2.4001} $\pm$ 0.0554 & 0.3052 $\pm$ 0.0289 \\
   \textsc{w/o TFI} & \underline{0.2422} $\pm$ 0.0263& 0.2078 $\pm$ 0.0204& \textbf{0.1631} $\pm$ 0.0173& & \underline{1.8688} $\pm$ 0.0586 &2.3730 $\pm$ 0.0542 & \textbf{0.3947} $\pm$ 0.0255 \\
   \bottomrule
   \end{tabular}}
       \vspace{-0.3cm}
\end{table*}

\subsubsection{Baselines and Metrics}
\textit{Baselines.} We compare our UMLC with \textbf{S-Learner}~\cite{kunzel2019metalearners}, \textbf{T-Learner}~\cite{kunzel2019metalearners}, \textbf{TARNet}~\cite{shalit2017estimating}, \textbf{CFRNet}~\cite{shalit2017estimating}, \textbf{DragonNet}~\cite{shi2019adapting}, \textbf{EUEN}~\cite{ke2021addressing} and \textbf{UniTE}~\cite{liu2023unite}. All of them are the representative methods in uplift modeling.
% \begin{itemize}
% \item \textbf{S-Learner}~\cite{kunzel2019metalearners}:  S-Learner is a kind of meta-learner method that uses a single estimator to estimate the outcome without giving the treatment a special role.
% \item \textbf{T-Learner}~\cite{kunzel2019metalearners}:  T-Learner is similar to S-Learner, which uses two estimators for the treatment and control groups, respectively. 
% \item \textbf{TARNet}~\cite{shalit2017estimating}: TARNet is a commonly used neural network-based uplift model. It uses the shared bottom network to extract feature information.
% \item \textbf{CFRNet}~\cite{shalit2017estimating}: CFRNet (\textit{i.e.}, CFRNet-wass, CFRNet-mmd) applies an additional loss to TARNet, which forces the learned treated and control feature distributions to be closer.
% \item \textbf{DragonNet}~\cite{shi2019adapting}: DragonNet exploits the sufficiency of the propensity score for estimation adjustment, and uses a regularization procedure based on the non-parametric estimation theory.
% \item \textbf{EUEN}~\cite{ke2021addressing}: EUEN is an explicit uplift modeling approach, which can correct the exposure bias.
% \item \textbf{UniTE}~\cite{liu2023unite}: UniTE adopts the Robinson Decomposition~\cite{bratteli2012operator} framework, and design a MMoE~\cite{ma2018modeling} based structure for uplift prediction.
% \end{itemize}

\noindent
\textit{Metrics.} Following the setup of previous work~\cite{belbahri2021qini}, we employ three evaluation metrics commonly used in uplift modeling, \textit{i.e.}, \textsc{AUUC (Area under Uplift Curve)}, \textsc{QINI (Qini Coefficient)} and \textsc{KENDALL (Kendall's Rank Correlation)}. 
% \begin{itemize}
% \item \textsc{AUUC (Area under Uplift Curve)}: A common metric to evaluate the area under the uplift curve ~\cite{rzepakowski2010decision}. We use the CausalML package~\cite{chen2020causalml} to implement the metric.
% \item \textsc{QINI (Qini Coefficient)}~\cite{mouloud2020adapting}: A common metric to evaluate the area under the qini curve, different from AUUC, it scale the responses in control group.
% \item \textsc{KENDALL (Kendall's Rank Correlation)}~\cite{mouloud2020adapting}: A metric to evaluate the average predicted uplift and the predicted uplift in each bin, we report the result of 20 bins. 
% \end{itemize}

\subsubsection{Implementation Details}
We implement all baselines and our UMLC based on Pytorch 1.10, with Adam as the optimizer and a maximum iteration count of 50. We use the QINI as a reference to search for the best hyper-parameters for all baselines and our model. We also adopt an early stopping mechanism with a patience of 5 to avoid over-fitting to the training set. Furthermore, we utilize the hyper-parameter search library Optuna~\cite{akiba2019optuna} to accelerate the tuning process, all experiments are implemented on NVIDIA A40 and Intel(R) Xeon(R) 5318Y Gold CPU @ 2.10GHz. 
%==
The code of our UMLC framework is provided in \url{https://github.com/ZexuSun/UMLC}.
\begin{figure*}[!t]
    \centering
    \includegraphics[width=0.8\linewidth]{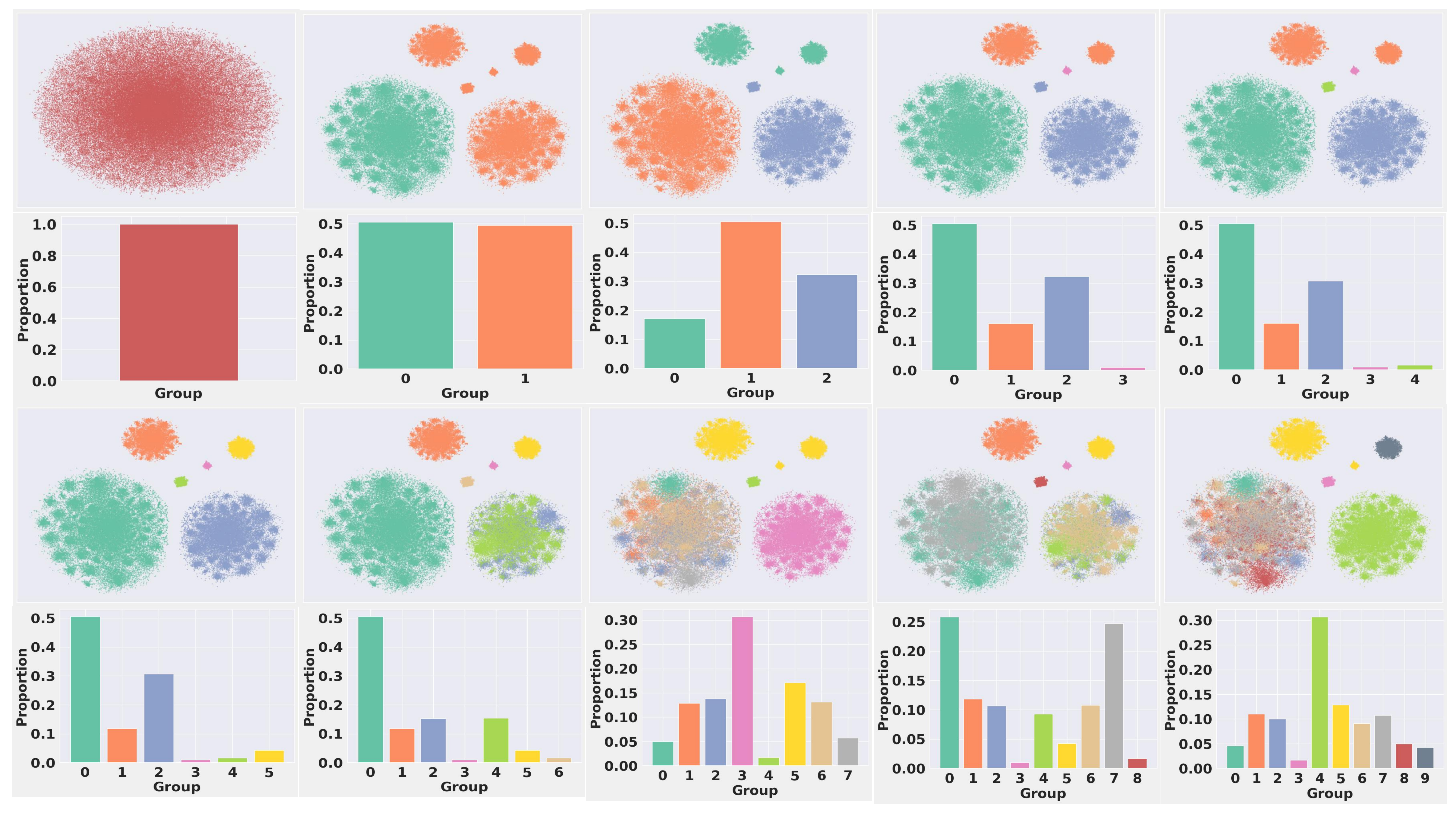}
    \vspace{-0.2cm}
    \caption{The context embedding t-SNE visualization of different group number $K$ (2-10) and the sample counts on the trained context embedding of Synthetic dataset. The first is the distribution of the original data. }
    \label{fig:cluster}
    \vspace{-0.2cm}
\end{figure*}
\subsection{Overall Performance (RQ1)}
To evaluate the effectiveness and compatibility of our UMLC, we conduct experiments on a synthetic dataset and a production dataset. The results are presented in Table~\ref{tab:overall}, and we have the following observations. 
%==
1) S-Learner and T-Learner get competitive results on the two datasets among all the baseline methods. 
%==
This indicates that the complex design of model structure (\textit{e.g.}, representation balancing, target regularization) may not be beneficial for the uplift prediction in real-time marketing. 
%==
In particular, the massive data of this scenario suffers from the significant distribution shift among the two groups; without the specific model structure design, the uplift models may get even worse uplift prediction results than the S-Learner and T-Learner.  
%==
2) CFRNet performs competitively among all the representation learning-based methods, but due to the representation balancing of Integral Probability Metrics (\textit{e.g.}, Maximum Mean Discrepancy, Wasserstein), the training time cost is high, especially the Wasserstein, which is optimized by the line search.
%==
The results of DragonNet, EUEN, and UniTE may be worse than the meta-learner methods; this may be because complex model structure design may not be useful for the uplift task.
%==
3) We combine our UMLC with four commonly used uplift modeling methods (\textit{i.e.}, CFRNet, DragonNet, EUEN, UniTE). 
%== 
We can find that integrating our UMLC on different uplift models can consistently achieve performance improvement, especially on QINI. 
%==
This may be because we tune the model hyperparameters by using the QINI as the objective.
%==
This demonstrates the effectiveness of our UMLC, where the carefully designed
two modules can effectively collaboratively reduce the distribution shift and perform feature interaction to improve the model’s performance.
%==
After combining the four base uplift models with our UMLC, the performance is improved across diverse metrics on the two datasets. 
%==
This suggests that our UMLC can serve as a general framework to improve the performance of uplift models in real-time marketing.
% \vspace{0.15cm}
%==

% \vspace{-0.5cm}
% % \vspace{0.2cm}
% \vspace{0.15cm}
\subsection{Ablation Study (RQ2)}

In this section, we conduct the ablation studies of our UMLC and
analyze the role played by each module. 
%==
We sequentially remove the components of the UMLC, \textit{i.e.}, the Response-guided Context Grouping (RCG), the User-Context Interaction (UCI) and the Treatment-Feature Interaction (TFI). 
%==
We construct three variants of UMLC, which are denoted as \textsc{w/o RCG}, \textsc{w/o UCI}, and \textsc{w/o TFI}. 
%==
We present the results in Table~\ref{tab:ablation}, and we can see
that removing any part may bring performance degradation.
%==
This verifies the validity of each part designed in our UMLC. 
%==
In particular, the response-guided context grouping module can narrow the contexts' feature space, reducing the distribution shift caused by the large-scale contexts.  
%==
The user-context interaction part introduces the relationship between the user and context features, which can help the uplift model better predict the user response.
%==
The treatment-feature interaction part can help the model discover the treatment assignment sensitive features, and the designed loss weight term can help the model adjust the sample importance, which can assign a bigger weight for the treatment assignment sensitive samples.
%==
All the components are helpful in solving the real-time marketing problem
and improving prediction performance.
\begin{figure}[!t]
    \centering
    \subfigure[Groups matching results]{\includegraphics[width=0.48\linewidth]{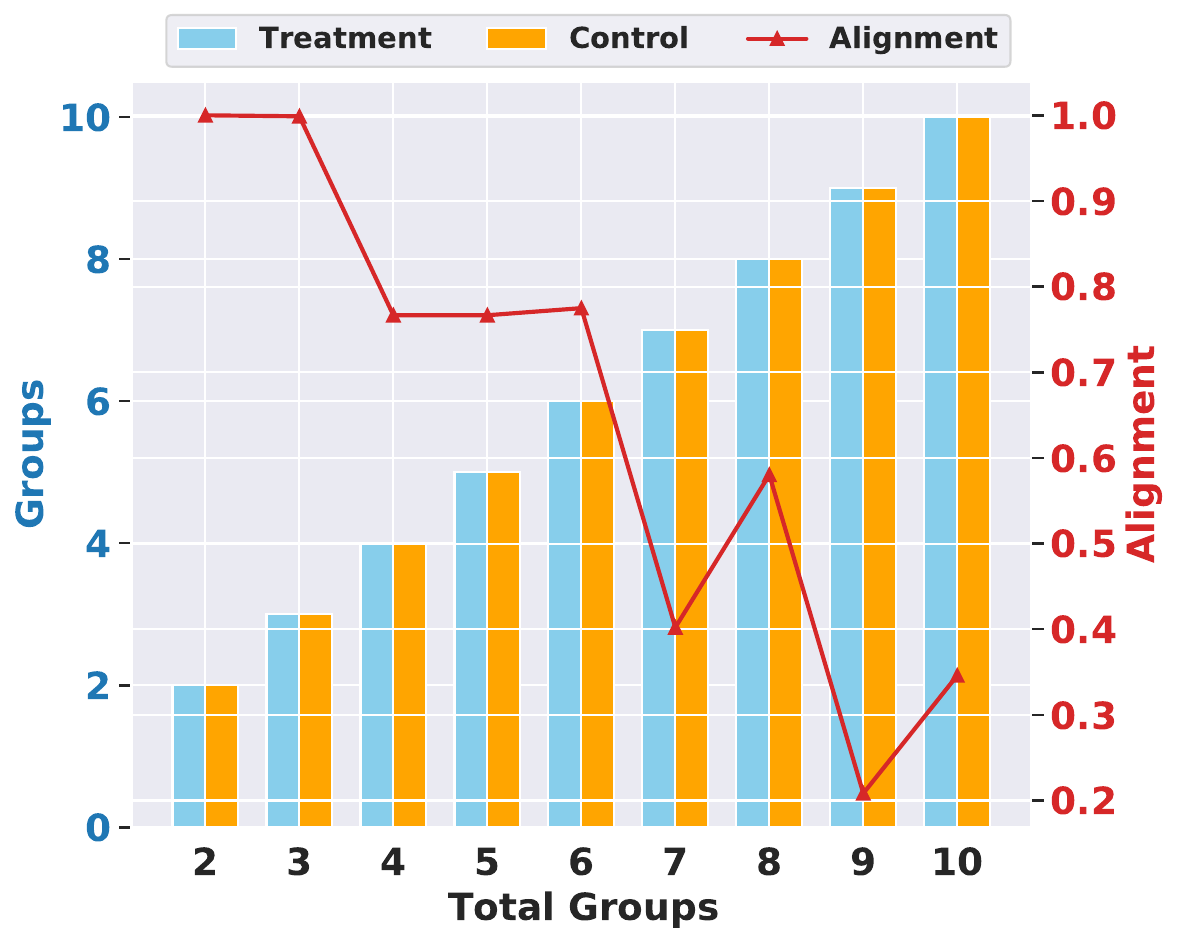}}\quad 
    \subfigure[Performance with different groups]{\includegraphics[width=0.48\linewidth]{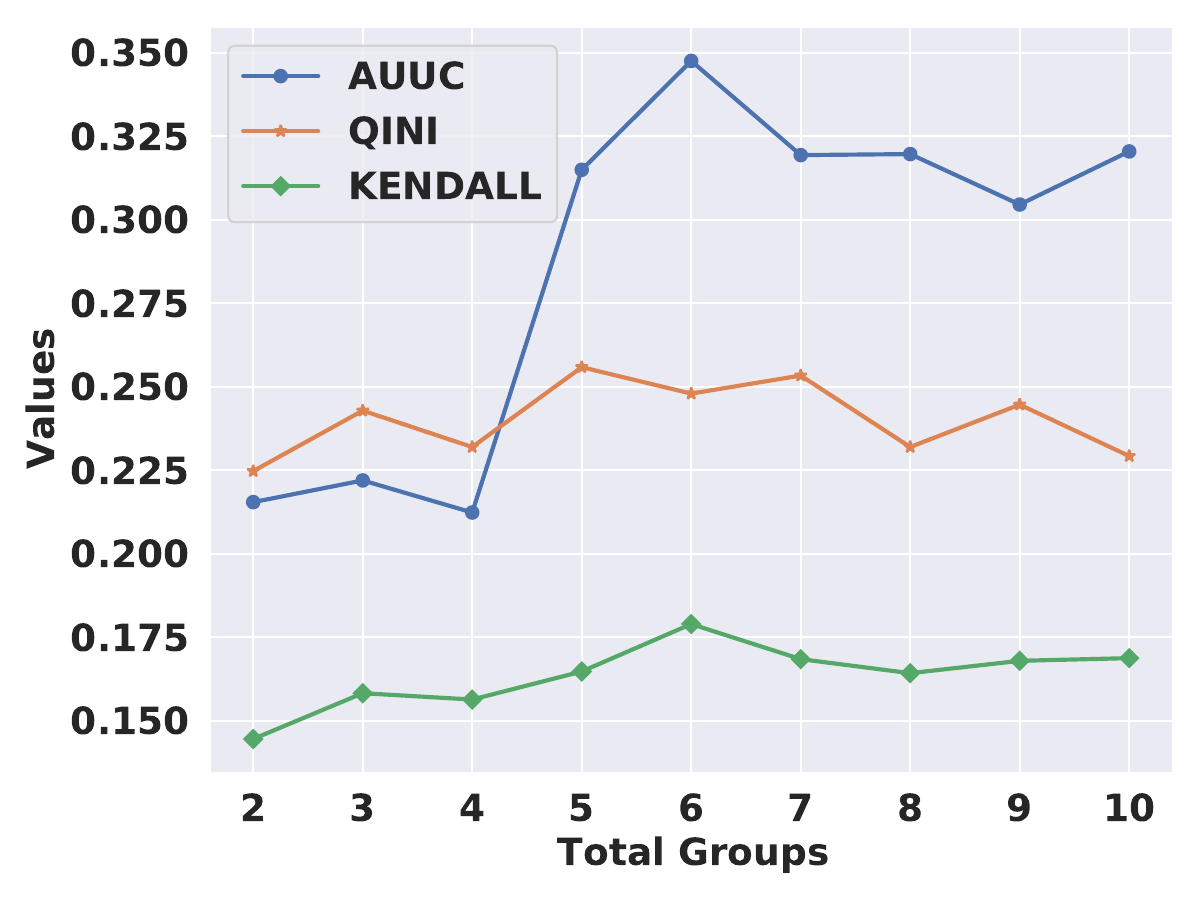}}
    \vspace{-0.5cm}
    \caption{The evaluation of different group number $K$ (2-10). We report the mean over five runs with different seeds.}
    \vspace{-0.5cm}
    \label{fig:clu_performance}
\end{figure}

% \begin{table}[!t]
% \centering
% \caption{Caption for the table.}\label{fig:clu_performance}
% \resizebox{\linewidth}{!}{
% \begin{tabular}{lccccccccccccc}
% \toprule
% \textbf{Group} & \textbf{2} & \textbf{3} & \textbf{4} & \textbf{5} & \textbf{6} & \textbf{7} & \textbf{8} & \textbf{9} & \textbf{10}  \\
% \midrule
% \textsc{Treatment} & 2&3 & 4 & 5 & 6 & 7 & 8 & 9 & 10\\
% \textsc{Control} & 2 & 3 & 4 & 5 & 6 & 7 & 8 & 9 & 10\\
% \textsc{Alignment} & 1.00 & 0.99 & 0.76 & 0.76 & 0.77 & 0.40 & 0.58 & 0.21 & 0.35\\
% \textsc{AUUC} & 0.28 & 0.32 & 0.32 & 0.31 & 0.32 & 0.29 & 0.26 & 0.22 & 0.24\\
% \textsc{QINI} & 0.23 & 0.25 & 0.24 & 0.26 & 0.26 & 0.25 & 0.27 & 0.23 & 0.19\\
% \textsc{KENDALL} & 0.13 & 0.15 & 0.13 & 0.14 & 0.16 & 0.16 & 0.17 & 0.14 & 0.15 \\
% \bottomrule
% \end{tabular}}
% \end{table}
\subsection{Context Grouping Analysis (RQ3)}
% We use the response-guided context grouping module for reducing the distribution shift introduced by the large-scale contexts.
% %==
% As mentioned in Section~\ref{sec:group}, the group number $K$ is an important hyperparameter, a proper group number $K$ can significantly improve the model performance.  

In this section, we detailed an analysis of the choice of $K$. 
%==
We show the cluster results of the trained context embedding on the Synthetic dataset in Figure~\ref{fig:cluster}.
%==
In the Synthetic dataset, we design six groups of contexts through the influence of the response, and from Figure~\ref{fig:cluster}, we can see that the trained context embedding can reflect the groups, even though the original data visualization has no evidence. 
%==
The cluster results show that when we select $K$ as 6, we can get a good grouping result. 
%==
To further illustrate the rationality of the group number choice, we conduct additional experiments in Figure~\ref{fig:clu_performance}.
%==
We define a metric Alignment to evaluate the performance of grouping, \textit{i.e.}, 
$$
\text{Alignment} = \frac{Count\left(\mathcal{M}\left((\boldsymbol{x}^{u_t},\boldsymbol{g}^t), (\boldsymbol{x}^{u_c},\boldsymbol{g}^c) \right)\right)}{Count(\boldsymbol{x}^{u_t})},
$$
% $$
% \text{Alignment} = \frac{Count(\boldsymbol{x}^{ut} \mid 
% \{ \boldsymbol{g} \mid \boldsymbol{x}^{ut}, \boldsymbol{x}^{uc}\} \in\{(\boldsymbol{g} \mid \boldsymbol{x}^{ut}) 
% \cap (\boldsymbol{g} \mid \boldsymbol{x}^{uc})\} )}{Count(\boldsymbol{x}^{ut})}
% $$
where $\boldsymbol{x}^{u_t}, \boldsymbol{g}^t$ and $\boldsymbol{x}^{u_c}, \boldsymbol{g}^c$ are the users and all the context groups in treatment and control groups, respectively. $\mathcal{M}(\cdot,\cdot)$ is to match the users based on the k-Nearest Neighbors (k-NN)~\cite{frolich2004finite}, then remain the users $\boldsymbol{x}^{u_t}$ where $\boldsymbol{g}^t=\boldsymbol{g}^c$, and the $Count(\cdot)$ represents the sample counter. 
%==
Alignment represents the proportion of similar users in treatment and control groups that match the same context groups of all the users in the treatment group.

From the results, we can see that although the group number in treatment and control samples is consistent with the total group number, the alignment is decreased with $K$ increases and suffers a drastic decline between 6 and 7.
%==
Specifically, we also show the model performance (\textit{i.e.}, UMLC (CFRNet-mmd)) across all the metrics with different group numbers. 
%==
The results show that the model performance increases on most metrics from $K=2$ to 6, and also decreases on the $K=7$. 
%==
This indicates that, although the group number is a hyperparameter, with a well trained response-guided context grouping module, we can target its value from the embedding space.

\section{Conclusion}
In this paper, to model the uplift in real-time marketing with massive data, we propose a novel model-agnostic Robust \underline{U}plift \underline{M}odeling with \underline{L}arge-Scale \underline{C}ontexts (UMLC) framework. 
%==
UMLC consists of two customized modules. 1) The response-guided context grouping module, which trains the context embedding through the regression model, then clusters the embedding into groups. 
%==
2) The feature interaction module, which contains two parts, \textit{i.e.}, the user-context interaction and the treatment-feature interaction. 
%==
We conduct extensive evaluations to validate the effectiveness of our UMLC, we also demonstrate the compatibility, which can be used as a general framework for real-time marketing.
\begin{acks}
We thank the support of the National Natural Science Foundation of China (No.62302310).
\end{acks}
% \newpage
\bibliographystyle{ACM-Reference-Format}
\bibliography{sample-base}


\appendix

% \section{Notation Table}

% For ease of understanding, we summarize a notation table of the important symbols in the paper.
% \begin{table}[h]
%     \centering
%     \caption{Notation Table. This table provides an overview of important symbols and their respective descriptions used in the paper.}\label{tab:symbol}
%     \begin{tabular}{c|l}\toprule
%     \textbf{Symbol}  & \textbf{Description} \\\midrule
%     $\boldsymbol{x}^u$ & User features \\
%     $\boldsymbol{x}^c$ & Context features \\
%     $t$  & Treatment \\
%     $y_i$ & User response \\
%     $\tau$  & Individual treatment effect or uplift. \\

%     $\boldsymbol{g}$  & Corresponding group variable. \\
%     $T$  & Treatment indicator variable. \\
%     $T$  & Treatment indicator variable. \\
%     $T$   & Treatment indicator variable. \\
%     $T$  & Treatment indicator variable. \\
% \\\bottomrule
%     \end{tabular}
% \end{table}
\newpage

\section{proof of Proposition \ref{prop1}}\label{app:proof}

\setcounter{proposition}{0} 
\setcounter{proposition}{0} 
\begin{proposition}
(\textbf{Restatement}) If we can find a good predictive function $f$ and transformation function $\xi$ with Lipschitz constraints on contexts such that $|h\left(\boldsymbol{x}^u, \boldsymbol{x}^c, t\right)-f\left(\boldsymbol{x}^u, \xi(\boldsymbol{x}^c), t\right)| \leq \mu, \forall \boldsymbol{x}^u,\boldsymbol{x}^c,t$ and $\left|f\left(\boldsymbol{x}^u, \xi(\boldsymbol{x}_i^c), t\right)\right.$ $\left.- f\left(\boldsymbol{x}^u, \xi(\boldsymbol{x}_j^c), t\right)\right| \leq c \cdot\left\|\xi(\boldsymbol{x}_i^c)-\xi(\boldsymbol{x}_j^c)\right\|_j, \forall \boldsymbol{x}^u,\boldsymbol{x}^c_i,$ $\boldsymbol{x}^c_j,t$. Then the function $\xi$ satisfies the Assumption~\ref{assm2} with $\zeta=c, \eta=2 \mu$.    
\end{proposition}
\begin{proof}
$$
\begin{aligned}
& \left|h\left(\boldsymbol{x}^u, \boldsymbol{x}^c_i, t\right)-h\left(\boldsymbol{x}^u,\boldsymbol{x}^c_j, t\right)\right| 
\leq\left|f\left(\boldsymbol{x}^u, \xi(\boldsymbol{x}_i^c), t\right)-f\left(\boldsymbol{x}^u, \xi(\boldsymbol{x}_j^c), t\right)\right| \\
& +\left|f\left(\boldsymbol{x}^u, \xi(\boldsymbol{x}_i^c), t\right)-h\left(\boldsymbol{x}^u, \boldsymbol{x}^c_i, t\right)\right|+\left|f\left(\boldsymbol{x}^u, \xi(\boldsymbol{x}_j^c)\right)-h\left(\boldsymbol{x}^u,\boldsymbol{x}^c_j, t\right)\right| \\
& \leq c \cdot\left\|\xi(\boldsymbol{x}_i^c)-\xi(\boldsymbol{x}_j^c)\right\|_2+2 \mu
\end{aligned}
$$
\end{proof}

\section{Computational Complexity}
In this section, we analyze the computational complexity of our UMLC. 
%==
The total computational complexity of the $k$-means, co-attention and cross-attention is $O\left(n K_l K_c\right)+O\left(n K_d^2\right)+O\left(\left(n+K_c\right) n K_u\right)+O\left(\left(K_c+n\right) K_u^2\right)+O\left(\left(K_u+n\right) K_c^2\right)$, where $n$ is the number of samples, $K$ is the number of clusters, $l$ is the iterations number of $K$ means, $K_c$ is the dimension of the context embedding, $K_u$ is the dimension of the user embedding, $K_d$ is the dimension of the concatenation of context and user embedding.

\section{Detailed Experiments}\label{app:exp}
\subsection{Synthetic Dataset Generation.}\label{app:syn}
In our problem setting, the synthetic data consists of user features $\boldsymbol{x}^u$; context features $\boldsymbol{x}^c$, treatment $t$, and response $y$, we describe the generation process in the following.

\noindent
$\bullet$ \textbf{User Features Generation.}
We assume that the user features is a $p$ dimensional vector $x^{u}_1, \ldots, x^{u}_p$. Specially, we generate them with the binary features $x^{u}_1, \ldots, x^{u}_{p_b} \sim \text{Bernoulli}(0.5)$, 
%==
and the continuous features $x^u_{p_b+1}, \ldots, x^u_{n_b+n_c} \sim \mathcal{N}(0, 1)$, where $p_b$ and $p_c$ are the dimensions of binary features and the continuous features, respectively, and $p=p_b+p_c$. We set $p_b=34$ and $p_c=66$.

\noindent
\textbf{Context features Generation.}
To simulate the real-word data, we use $q$ as the length of the context features, and there are three types of features in the context features $\boldsymbol{x}^c$, \textit{i.e.}, binary features, continuous features and categorical features. 
%==
Specially, the binary features $x^c_{1}, \ldots, x^c_{q_b} \sim \text{Bernoulli}(0.5)$, 
the continuous features $x^c_{q_b+1}, \ldots, x^c_{q_b+q_c}$ $\sim \mathcal{N}(0, 1)$
and the categorical features $x^c_{q_b+q_c+1}, \ldots, x^{c}_{q_b+q_c+q_m}$ $\sim \text{Multinomial}(4, 0.25, 0.25, 0.25,$ $0.25)$, 
We denote the dimensions as $q_b$, $q_c$, and $q_m$ for them. And we set $q_b=34$, $q_c=66$ and $q_m=3$.
%==
Referring to the Assumption ~\ref{assum1}, the influence of the context features should be divided into $K$ groups. 
%==
In the Synthetic dataset, we set $K=6$, and we generate another categorical feature $x^{c}_q\sim  \text{Multinomial}(6, 0.2, 0.16, 0.16, 0.16, 0.16, 0.16)$, whose influence for the response is separately modeled in the response generation.
\begin{figure}[!t]
    \centering
    \subfigure[Continuous feature 1]{\includegraphics[width=0.45\linewidth]{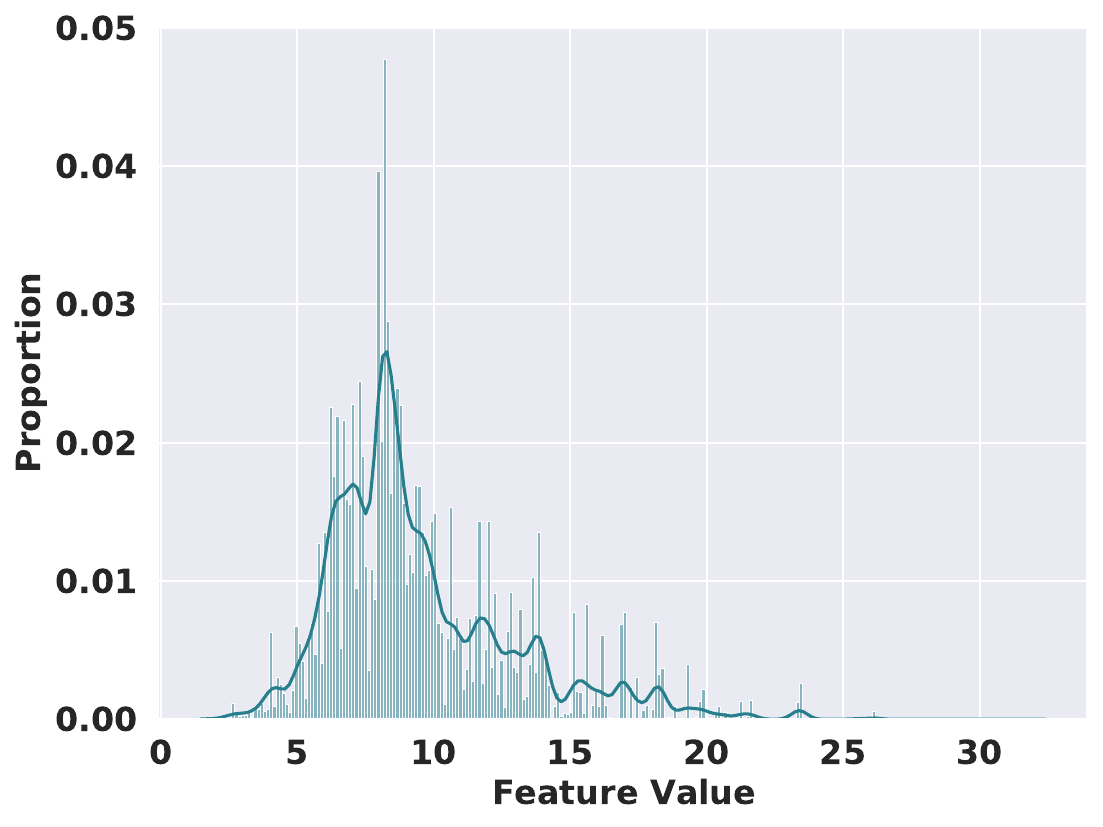}}
    \subfigure[Continuous feature 2]{\includegraphics[width=0.45\linewidth]{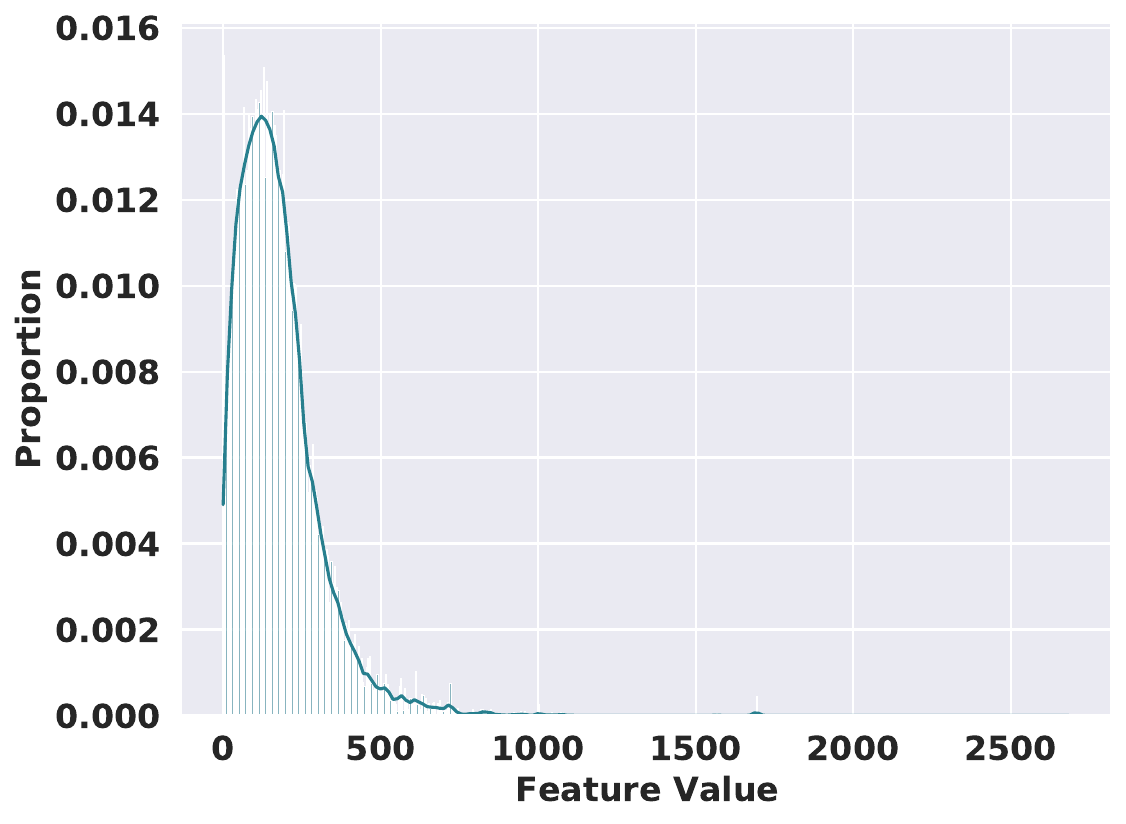}}\\
    \subfigure[Continuous feature 3]{\includegraphics[width=0.45\linewidth]{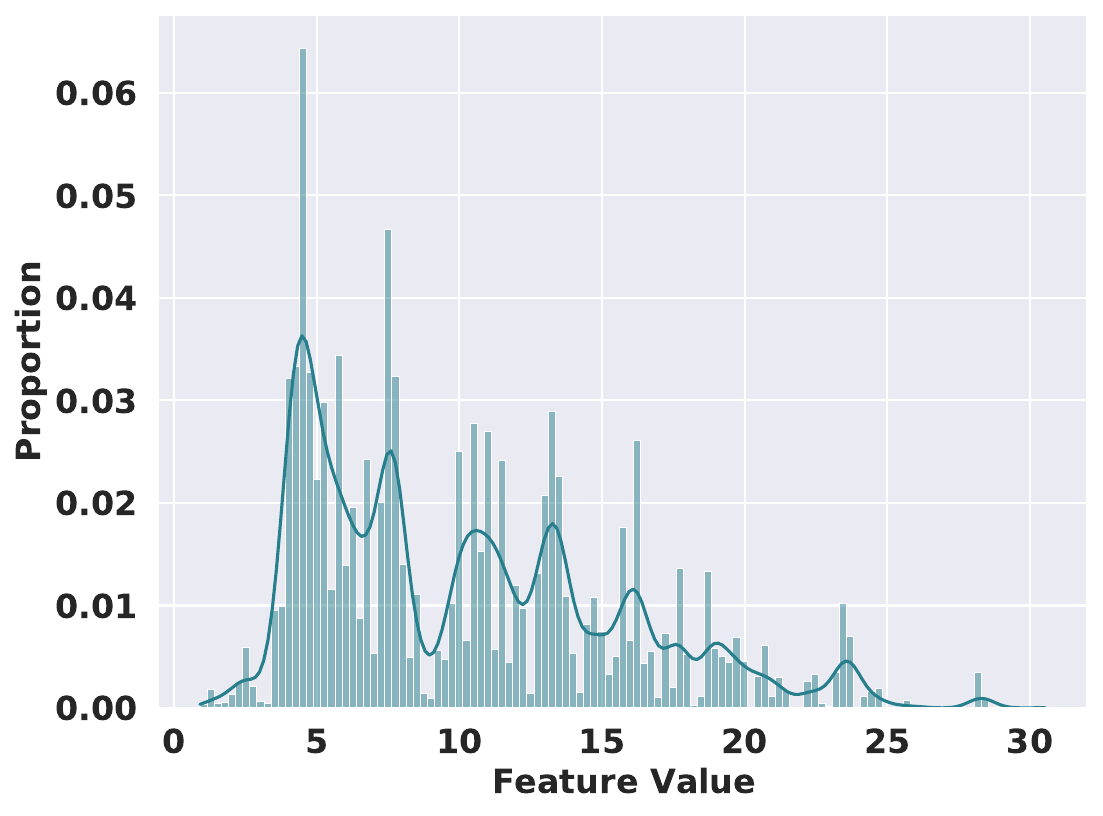}}
    \subfigure[Categorical feature]{\includegraphics[width=0.45\linewidth]{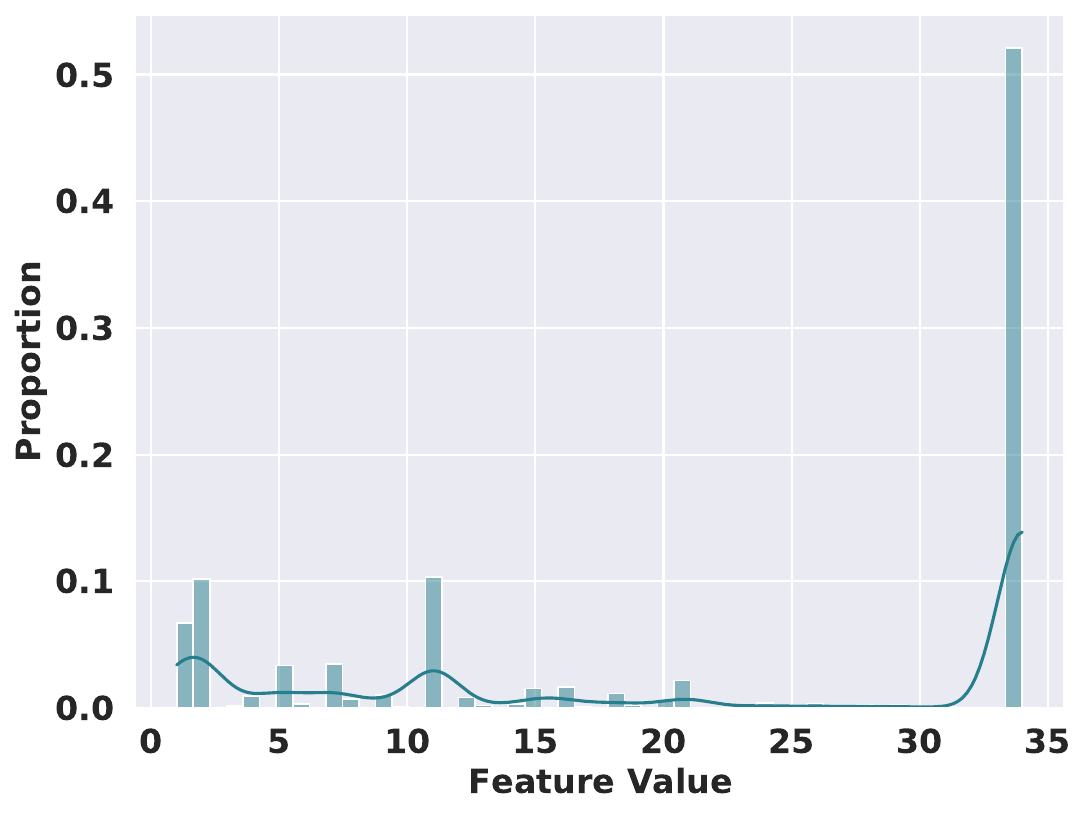}}
    \caption{The feature visualization of the Production dataset.}
    \label{fig:pro_visual}
\end{figure}
\noindent
\textbf{User-Context Features Combination.}
Our goal is to replicate the real-world, real-time marketing scenario, where a relatively small number of users engage with a significantly larger number of contexts. 
%==
After we generate the user features and the context features, we need to design the combination of the user and the contexts.
%==
Specifically,  we generate a large context pool (\textit{i.e.}, 100 times larger than the user number), and each user is associated with a uniquely and randomly chosen subset (\textit{i.e.}, 60-130 contexts) of these contexts, highlighting the one-to-many relationship fundamental to user-context interactions. 

\noindent
\textbf{Treatment Generation.}
In the randomized control trials, the treatment is randomly assigned to each user. Thus, we generate $t\sim  \text{Bernoulli}(0.5)$ for each user, and the samples belonging to each user are assigned with the same $t$.

%The final outcome $L$ is contingent on $T$ and is derived from $Y_0$ and $Y_1$ to reflect the observed results in the experimental and control groups, respectively. This methodology facilitates a comprehensive evaluation of the interplay between feature combinations, category-specific effects, and treatment impacts on the final outcome.

\noindent
\textbf{Response Generation.}
For each sample, we generate the two potential responses $y^0$ and $y^1$ as follows:  
\begin{equation*}
\begin{aligned}
	y^0 =& 0.5   \sum_{i=1}^{p} x^u_i + 0.5  \sum_{j=1}^{q_c+q_b} x^c_j+ 0.5  \sum_{j=1}^{q_b+q_c}( \sum_{i=1}^{p} x^u_i)  x^c_j\\&+ 0.5   \sum_{v=q_b+q_c+1}^{q_b+q_c+q_m}x^c_ v + z_0 +\epsilon_0,
\end{aligned}
\end{equation*}
where the distribution of $z_0$, \textit{i.e.}, $\mathcal{P}(z_0) \in \{\mathcal{N}(0,1), \mathcal{N}(2,0.5), \mathcal{N}(-1,2),$ $\mathcal{N}(3,1.5), \mathcal{N}(-2,0.8), \mathcal{N}(1,2)  \}$ for the corresponding group (\textit{i.e}., ${x}^c_d$) of the context, $\epsilon_0 \sim \mathcal{N}(0, 1)$ is the noise term.

Expanding upon $y^0$, the user response $y^1$ incorporates the effect of the treatment and is formulated as:
\begin{equation*}
\begin{aligned}
y^1 = &y^0 + 0.2\sum_{i=1}^{p} x^u_i + 0.2  \sum_{j=1}^{q_c+q_b} x^c_j+ 0.2  \sum_{j=1}^{q_b+q_c}( \sum_{i=1}^{p} x^u_i)  x^c_j\\
&+ 0.2   \sum_{v=q_b+q_c+1}^{q_b+q_c+q_m}x^c_ v + z_1 +\epsilon_1.     
\end{aligned}
\end{equation*}
where $z_1$ is similar to $z_0$, and $\epsilon_1 \sim \mathcal{N}(0, 1)$ is the noise term.

%==
Moreover, to better understand the two datasets that we used, we present the statistics of the two datasets in Table~\ref{tab:data}.
\begin{table}[!t]
	\centering
	\caption{Statistics of the Synthetic dataset and Production Dataset. We randomly split the two datasets for train/validation/test split proportion of 70\%/20\%/10\%.}
	\begin{tabular}{ccccc}\toprule
		\multirow{2}{*}{\textbf{Dataset}}  & \multicolumn{2}{c}{\textbf{Feature number}} &\multicolumn{2}{c}{\textbf{User number}}\\ 
		& User & Context &Treated &Control\\\midrule
		Synthetic &100 &103 & 236421 & 236116 \\
		Production &66 &109 &397943  & 397852 \\ \bottomrule
	\end{tabular}
	\label{tab:data}
\end{table}

\subsection{Production Dataset Visualization}
In this section, we make some visualization for the Production dataset, the results are shown in Figure~\ref{fig:pro_visual}, which can help the readers make decisions to implement our UMLC on their own scenario.

\subsection{Experimental Setups}
We present the detailed description of the baselines we used in the following:
\begin{itemize}
	\item \textbf{S-Learner}~\cite{kunzel2019metalearners}:  S-Learner is a kind of meta-learner method that uses a single estimator to estimate the outcome without giving the treatment a special role.
	\item \textbf{T-Learner}~\cite{kunzel2019metalearners}:  T-Learner is similar to S-Learner, which uses two estimators for the treatment and control groups, respectively. 
	\item \textbf{TARNet}~\cite{shalit2017estimating}: TARNet is a commonly used neural network-based uplift model. It uses the shared bottom network to extract feature information.
	\item \textbf{CFRNet}~\cite{shalit2017estimating}: CFRNet (\textit{i.e.}, CFRNet-wass, CFRNet-mmd) applies an additional loss to TARNet, which forces the learned treated and control feature distributions to be closer.
	\item \textbf{DragonNet}~\cite{shi2019adapting}: DragonNet exploits the sufficiency of the propensity score for estimation adjustment, and uses a regularization procedure based on the non-parametric estimation theory.
	\item \textbf{EUEN}~\cite{ke2021addressing}: EUEN is an explicit uplift modeling approach, which can correct the exposure bias.
	\item \textbf{UniTE}~\cite{liu2023unite}: UniTE adopts the Robinson Decomposition~\cite{bratteli2012operator} framework, and design a MMoE~\cite{ma2018modeling} based structure for uplift prediction.
\end{itemize}

We present the detailed description of the metrics we used in the following:
\begin{itemize}
	\item \textsc{AUUC (Area under Uplift Curve)}: A common metric to evaluate the area under the uplift curve ~\cite{rzepakowski2010decision}. We use the CausalML package~\cite{chen2020causalml} to implement the metric.
	\item \textsc{QINI (Qini Coefficient)}~\cite{mouloud2020adapting}: A common metric to evaluate the area under the qini curve, different from AUUC, it scale the responses in control group.
	\item \textsc{KENDALL (Kendall's Rank Correlation)}~\cite{mouloud2020adapting}: A metric to evaluate the average predicted uplift and the predicted uplift in each bin, we report the result of 20 bins. 
\end{itemize}
Additionally, we also evaluate the results of the ground truth uplift prediction on the Synthetic dataset, we use the following metrics:
\begin{itemize}
    \item $\epsilon_{ATE}$: The absolute error on the average treatment effect, which is defined as:
$$
\epsilon_{A T E}=\left|\frac{1}{n} \sum_{i=1}^n\left[y^1-y^0\right]-\frac{1}{n} \sum_{i=1}^n\left[\hat{y}^1-\hat{y}^0\right]\right|
$$
\item $\epsilon_{PEHE}$: In order to measure the accuracy of the uplift prediction, we use the Precision in Estimation of Heterogeneous Effect (PEHE), which is defined as:
$$
\epsilon_{P E H E}=\frac{1}{n} \sum_{i=1}^n\left[\left(y^1-y^0\right)-\left(\hat{y}^1-\hat{y}^0\right)\right]^2
$$
\end{itemize}
\begin{figure}[!t]
	\centering
	\includegraphics[width=\linewidth]{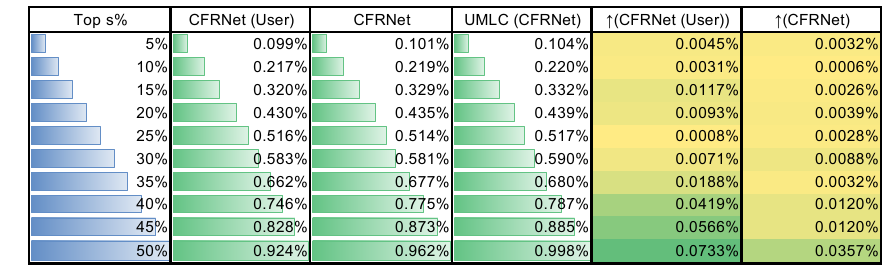}
	\caption{The results of uplift gains (\textit{i.e}., only user, user + context, our UMLC with user + context), we use 5\% as the step size.}
	\label{fig:uplift_gain}
\end{figure}
\begin{figure*}[!t]
	\centering
	\subfigure[Batch  size $b$ on Synthenic dataset]{\includegraphics[width=0.24\linewidth]{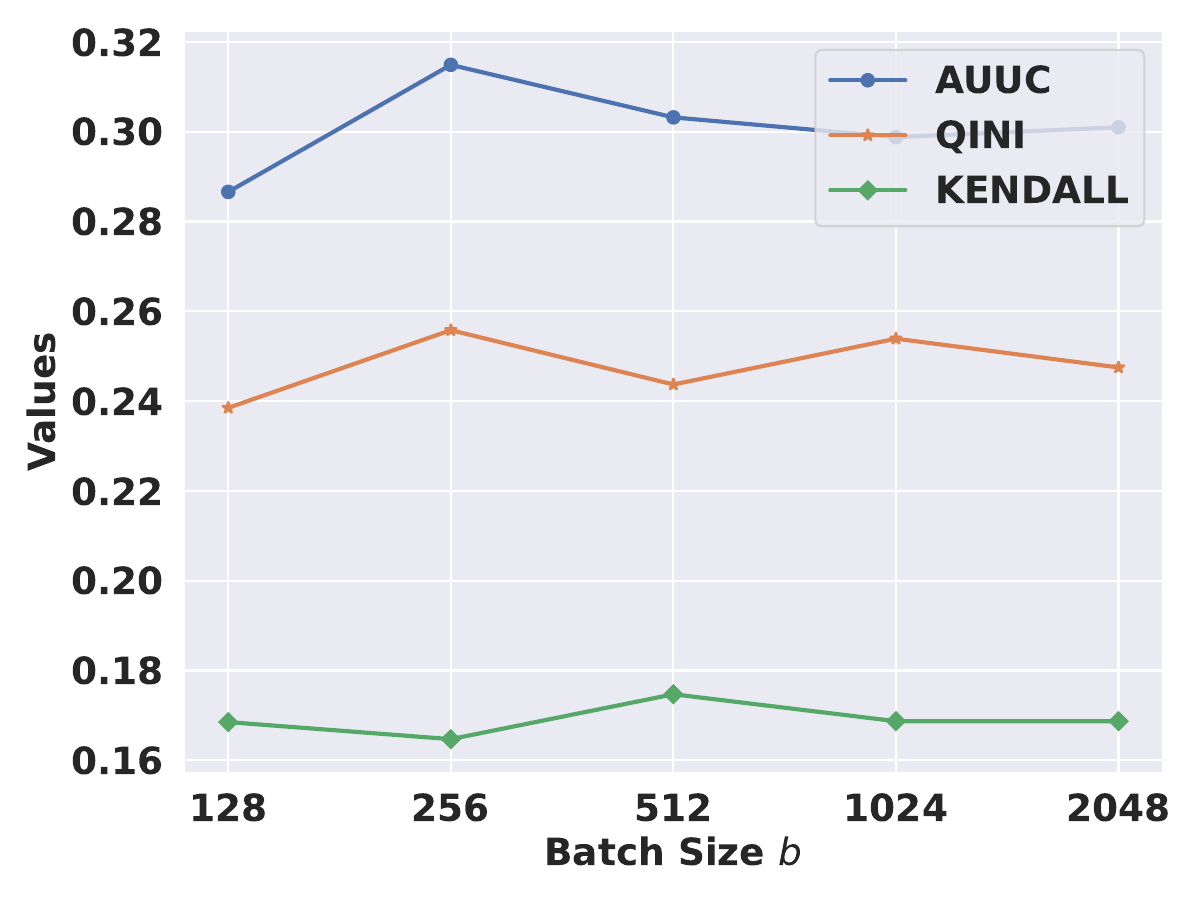}} 
	\subfigure[Loss weight $\beta$ on Synthenic dataset]{\includegraphics[width=0.24\linewidth]{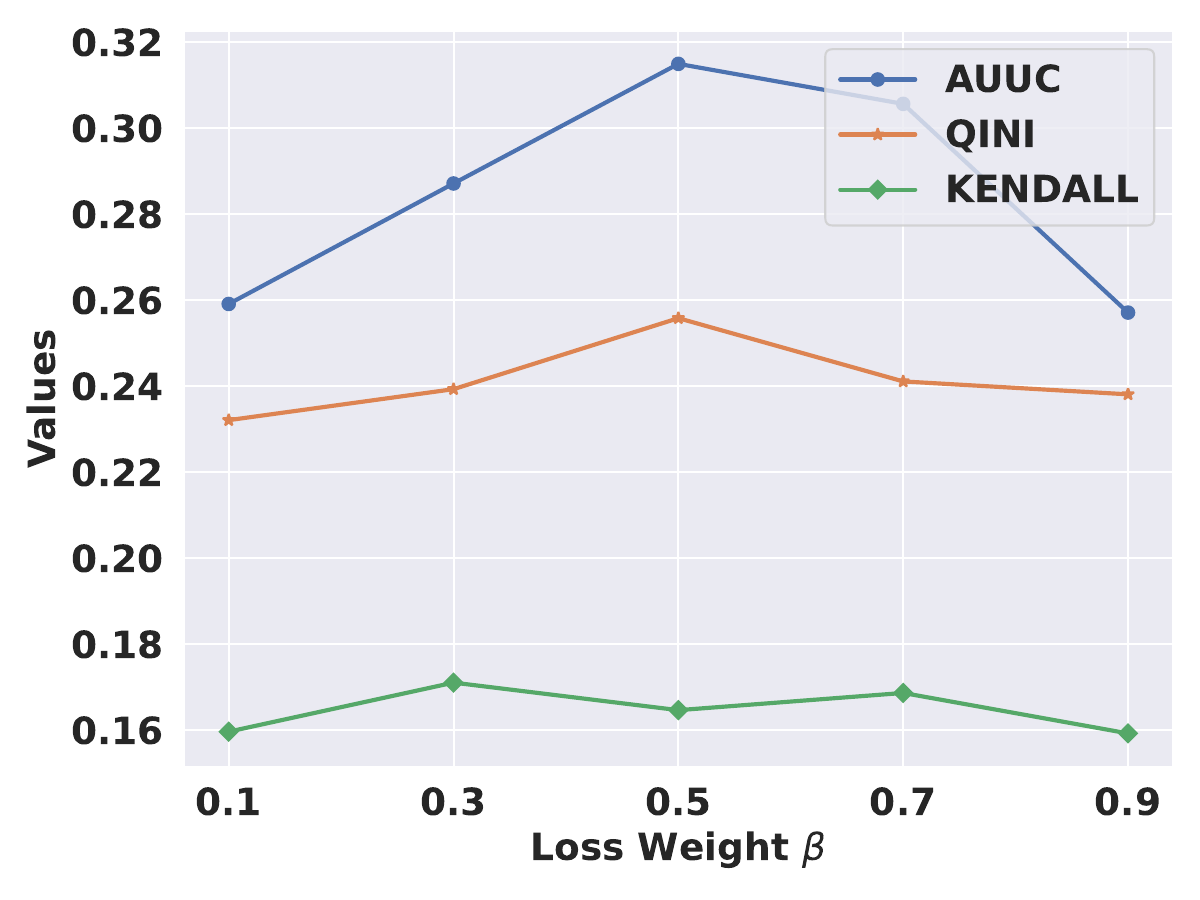}}
	\subfigure[Loss weight $\gamma$ on Synthenic dataset]{\includegraphics[width=0.24\linewidth]{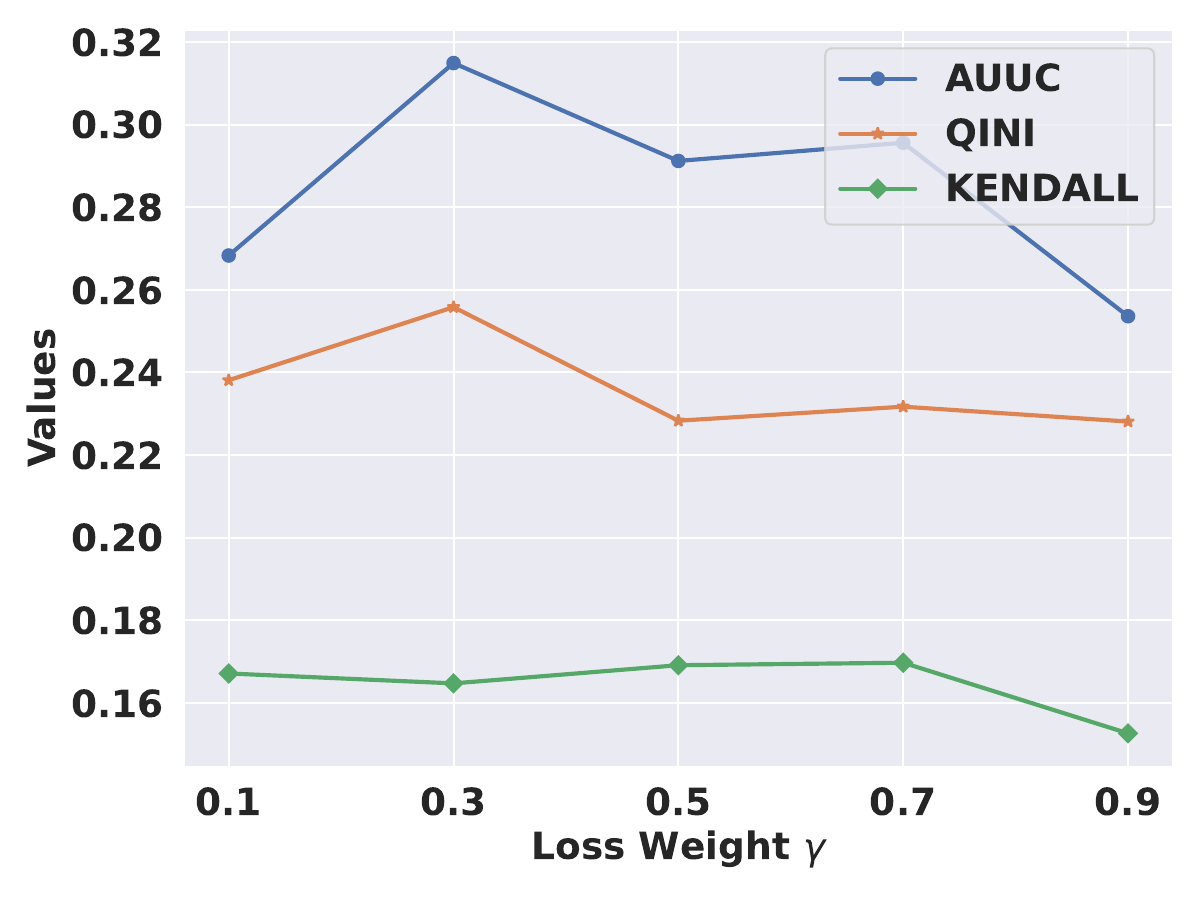}}\\
	\subfigure[Batch size $b$ on Production dataset]{\includegraphics[width=0.24\linewidth]{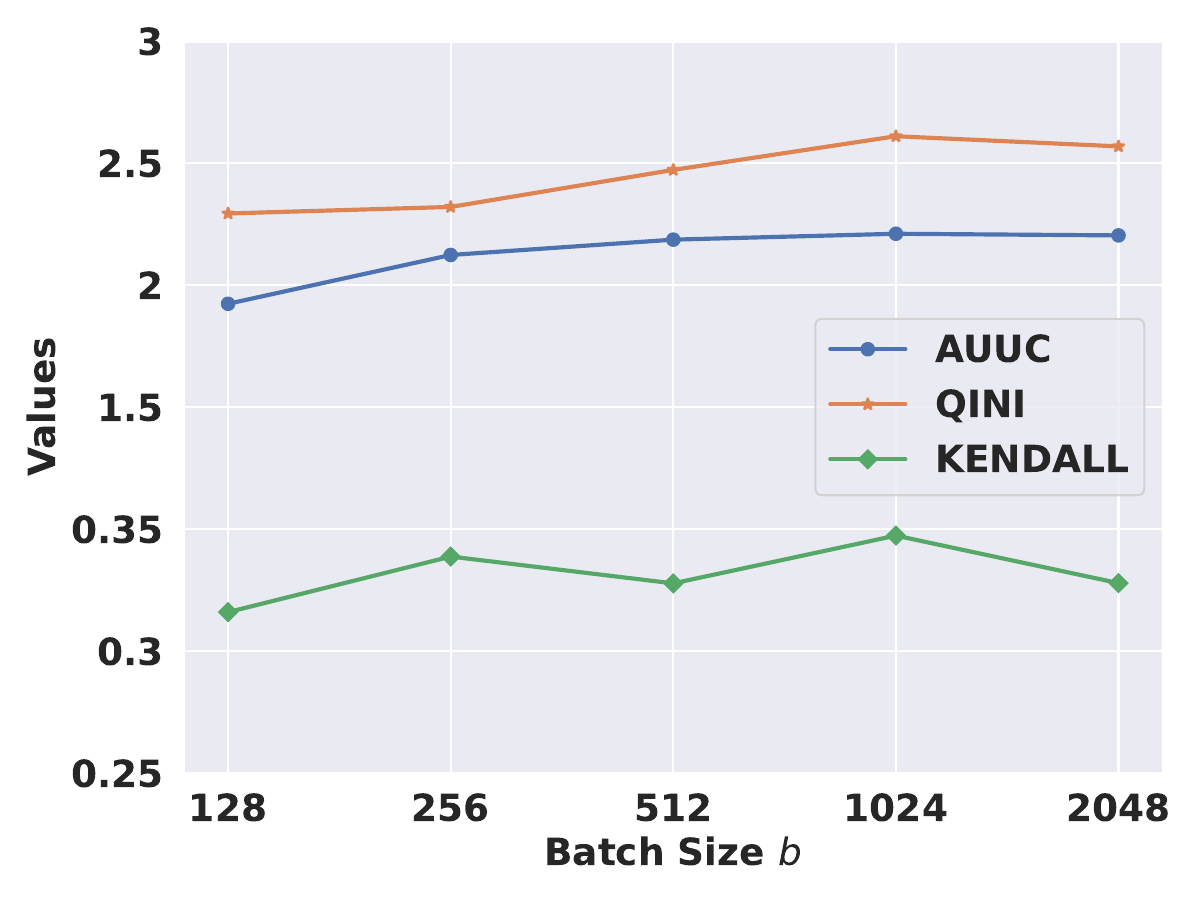}} 
	\subfigure[Loss weight $\beta$ on Production dataset]{\includegraphics[width=0.24\linewidth]{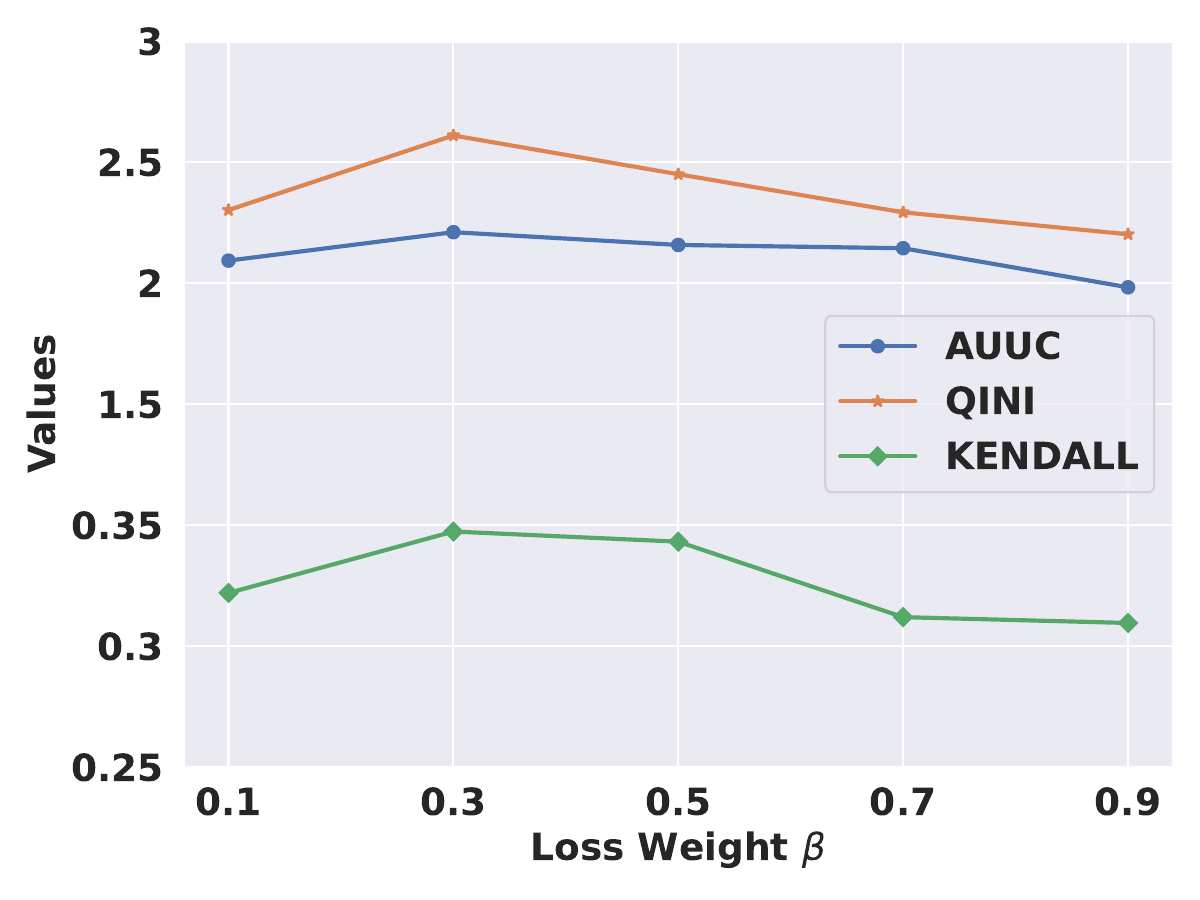}}
	\subfigure[Loss weight $\gamma$ on Production dataset]{\includegraphics[width=0.24\linewidth]{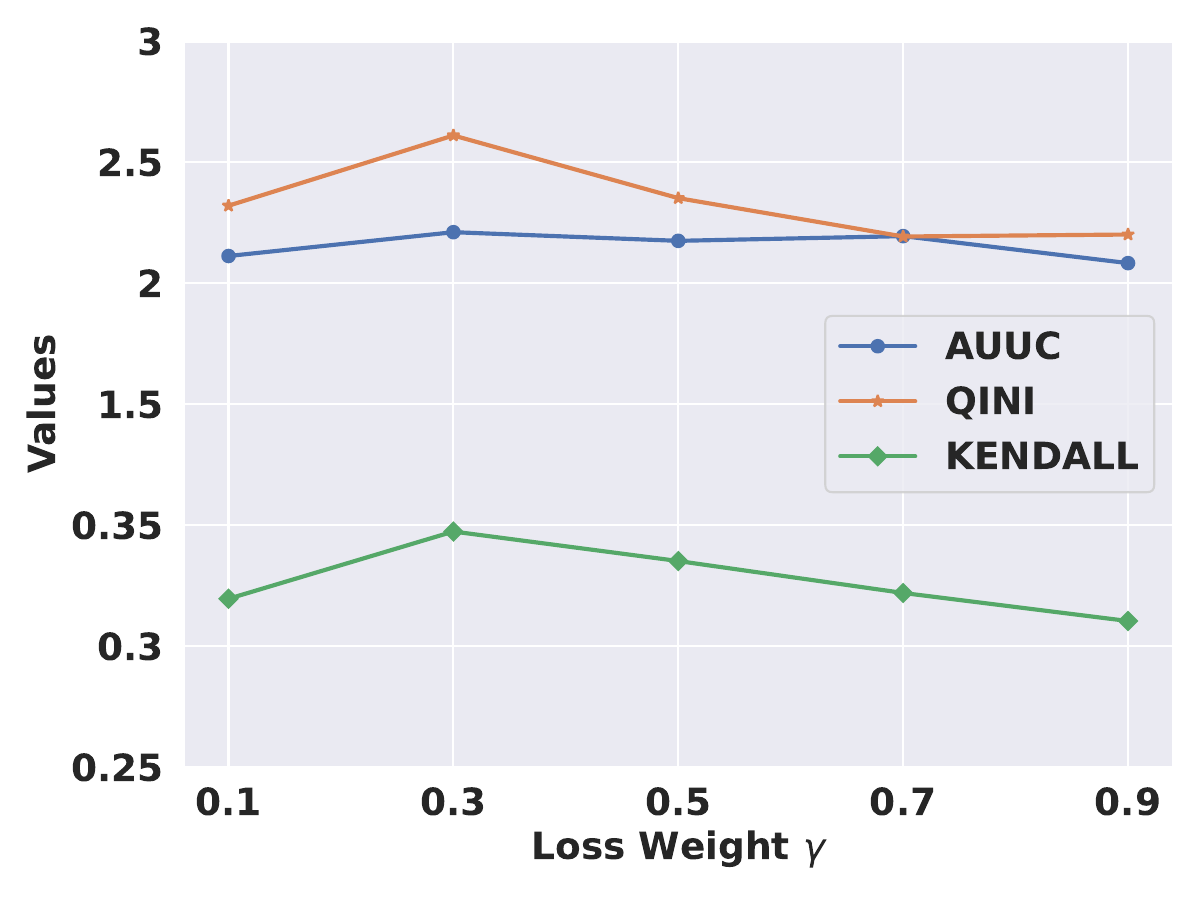}}
  \vspace{-0.3cm}
	\caption{The results of the hyperparameter analysis on the two datasets (\textit{i.e.}, batch size $b$, loss weight $\beta$ and loss weight $\gamma$). We report the results of the mean over five runs with different seeds.}
	\label{fig:hyper}
 \vspace{-0.2cm}
\end{figure*}

\subsection{Ground Truth Uplift Evaluation}

We present the evaluation of the ground truth uplift in Table~\ref{tab:ground_truth}. From the table, it can be observed that for $\epsilon_{A T E}$, UMLC (DragonNet) achieves the best performance with a mean value of 0.4083 , closely followed by UMLC (EUEN) with a mean of 0.4193 . For $\epsilon_{ PEHE }$, UMLC (EUEN) performed the best with a mean value of 2.3237 , while UMLC (CFRNet-mmd) ranked second with a mean of 2.3482. 
%==
Our UMLC consistent superiors performance across two metrics, which demonstrates UMLC's advantages and reliability in uplift prediction.
\begin{table}[!t]
\centering
\caption{The results of ground truth uplift evaluation on Synthetic Dataset. We report the results over five random seeds. The best results and second best results are \textbf{bold} and \underline{underlined}, respectively.}\label{tab:ground_truth}
\begin{tabular}{c|cc} \toprule
\textbf{Method} & $\epsilon_{ATE}$ & $\epsilon_{PEHE}$ \\
\midrule S-Learner & 0.6093 $\pm$ 0.1592 & 3.1983 $\pm$ 0.1921 \\
 T-Learner & 0.5727 $\pm$ 0.1628 & 2.9817 $\pm$ 0.1729 \\
 TARNet & 0.5139 $\pm$ 0.1529 & 2.7789 $\pm$ 0.1948 \\
 CFRNet-mmd & 0.5133 $\pm$ 0.1190 & 2.3743 $\pm$ 0.1829 \\
 CFRNet-wass & 0.5017 $\pm$ 0.1739 & 2.4283 $\pm$ 0.1638 \\
 DragonNet & 0.5129 $\pm$ 0.1294 & 2.7891 $\pm$ 0.1384 \\
 EUEN & 0.4938 $\pm$ 0.1599 & 2.3802 $\pm$ 0.2193 \\
 UniTE & 0.5213 $\pm$ 0.1074 & 2.8345 $\pm$ 0.1739 \\
\midrule  \rowcolor{mygray} \textbf{UMLC} (CFRNet-mmd) & 0.4932 $\pm$ 0.1439 & \underline{2.3482} $\pm$ 0.1274 \\
  \rowcolor{mygray} \textbf{UMLC} (CFRNet-wass) & 0.4721 $\pm$ 0.1548 & 2.3793 $\pm$ 0.1539 \\
  \rowcolor{mygray}  \textbf{UMLC} (DragonNet) & \textbf{0.4083} $\pm$ 0.1283 & 2.4930 $\pm$ 0.1983 \\
  \rowcolor{mygray} \textbf{UMLC} (EUEN) & \underline{0.4193} $\pm$ 0.1492 & \textbf{2.3237} $\pm$ 0.1385 \\
 \rowcolor{mygray} \textbf{UMLC} (UniTE) & 0.4892 $\pm$ 0.1947 & 2.6941 $\pm$ 0.1928 \\ \bottomrule
\end{tabular}
\end{table}

\subsection{Algorithm Pseudocode}\label{app:code}
In this section, to easily understand our proposed framework, we provide the pseudocode of our UMLC in Algorithm~\ref{alg:alg1}.
\begin{algorithm}[h]
\centering
\caption{The training procedure of our UMLC.}\label{alg:alg1}
\begin{algorithmic}[1]
\Require User features $\boldsymbol{x}^u$, context features $\boldsymbol{x}^c$, treatment $t$, response $y$, regression model parameters $\boldsymbol{\varphi}$, uplift model parameters $\boldsymbol{\omega}$, hyperparameters $\alpha$, $\beta$, $\gamma$.
\Ensure The regression model parameters $\boldsymbol{\varphi}$ and uplift model parameters $\boldsymbol{\omega}$.
\Statex \hspace{-0.5cm}{\textsc{Response-guided Context Grouping}}
\State Calculate the prediction loss according to Eq.~\eqref{eq:pred_loss}. 
\State Calculate the Lipschitz Regularization according to Eq.~\eqref{eq:lips}.
\State Update the regression model parameters $\boldsymbol{\varphi}$ according to Eq.~\eqref{eq:reg_loss}.
\State Group and aggregate the training data according to Eq.~\eqref{eq:group_data}.
\Statex \hspace{-0.5cm} \textsc{Feature Interaction}
\State Get the user and context attention matrix according to Eq.~\eqref{eq:att_matrix}.
\State Calculate the user and context attention vectors according to Eq.~\eqref{eq:atten_vec}. 
\State Predict the control response according to Eq.~\eqref{eq:pred_control}.
\State Calculate the information gain by treatment assignment according to Eq.~\eqref{eq:infor_gain}.
\State Predict the uplift according to Eq.~\eqref{eq:uplift_pred}.
\State Update the uplift model parameters $\boldsymbol{\omega}$ according to Eq.~\eqref{eq:uplift_loss}.
\end{algorithmic}
\end{algorithm}

\subsection{Uplift Gain Experiments}
To further evaluate the performance of our UMLC, we conduct the uplift gain experiments on the Synthetic dataset, which has two potential responses for each sample. 
%==
After we have the learned uplift models, we sort the samples by the uplifts that the model predicted. 
%==
And then, we assign the top $s\%$ of the samples with treatment $t=1$ and the rest with $t=0$.
%==
Comparing to the strategy that all the samples are assigned with $t=0$, we can calculate the uplift gain (\textit{i.e.}, Gain) by the sum of corresponding user responses, \textit{i.e.}, 
$$
\text{Gain} = \frac{ \sum_{s\%} y^1 +\sum_{\text{rest}} y^0-\sum y^0}{\sum y^0}\times 100 \%
$$
where ${s\%}$ and rest represent the top $s\%$ of the samples and the rest of the samples, respectively.
% In the context of evaluating different models' performance, the assessment hinges on the economic impact of decisions guided by model predictions. Specifically, we concentrate on the revenue generation potential of the top $s\%$ of records as identified by the models, where these selected records are assumed to be actioned upon with a strategy denoted by $t=1$. This strategy, $t=1$, might represent a high-confidence action such as investing more resources or targeting more aggressively, expected to yield higher revenue compared to a default or conservative action represented by $t=0$. Revenue under these two conditions is calculated differently to reflect the variance in expected outcomes based on model confidence.

% This approach is measured against a baseline where a conservative strategy $t=0$, implying lower expected revenue, is applied to all records. The revenue improvement from using the model's predictions is calculated with the formula:
% $$
% \text{Improvement} = \left( \frac{R_{s\%, t=1} + R_{\text{rest}, t=0} - R_{\text{all}, t=0}}{R_{\text{all}, t=0}} \right) \times 100 \%
% $$
% where $R_{s\%, t=1}$ is the revenue from the top $s\%$ of data with condition $t=1$; $R_{\text{rest}, t=0}$ is the revenue from the remaining data with condition $t=0$; $R_{\text{all}, t=0}$ is the baseline revenue when condition $t=0$ applies to all data. 

We present the results in Figure~\ref{fig:uplift_gain}, we use 5\% as the step size. From the results, we have the following observations.
%==
1) From the results of CFRNet (User) and CFRNet, we can see that the CFRNet performs better than CFRNet (User).
%==
This indicates that only considering the user to assign treatments for each sample may introduce bias for the model learning. 
%==
2) Our UMLC consistently outperforms CFRNet (User) and CFRNet across all the top $s\%$ of the samples. 
%==
With the percentage of the sample increase, the improvement of our UMLC becomes larger.
%==
This may be because when the $s\%$ increases, the context groups of our UMLC can mitigate the prediction variance of the large-scale contexts, which can help the model get a better ranking of the samples to achieve higher uplift gain.

\begin{figure}[!t]
	\centering
	\subfigure[Groups matching results]{\includegraphics[width=0.48\linewidth]{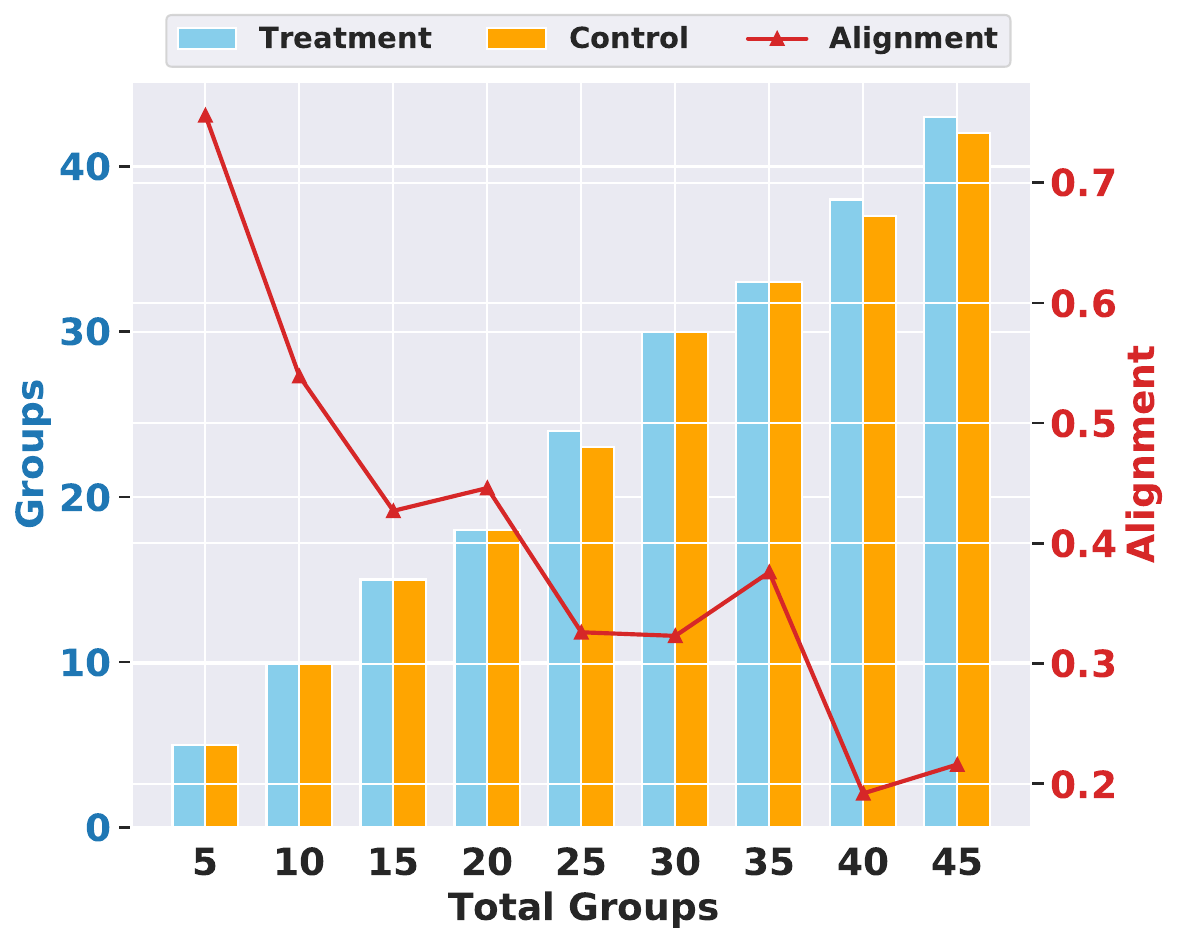}}\quad 
	\subfigure[Performance with different groups]{\includegraphics[width=0.48\linewidth]{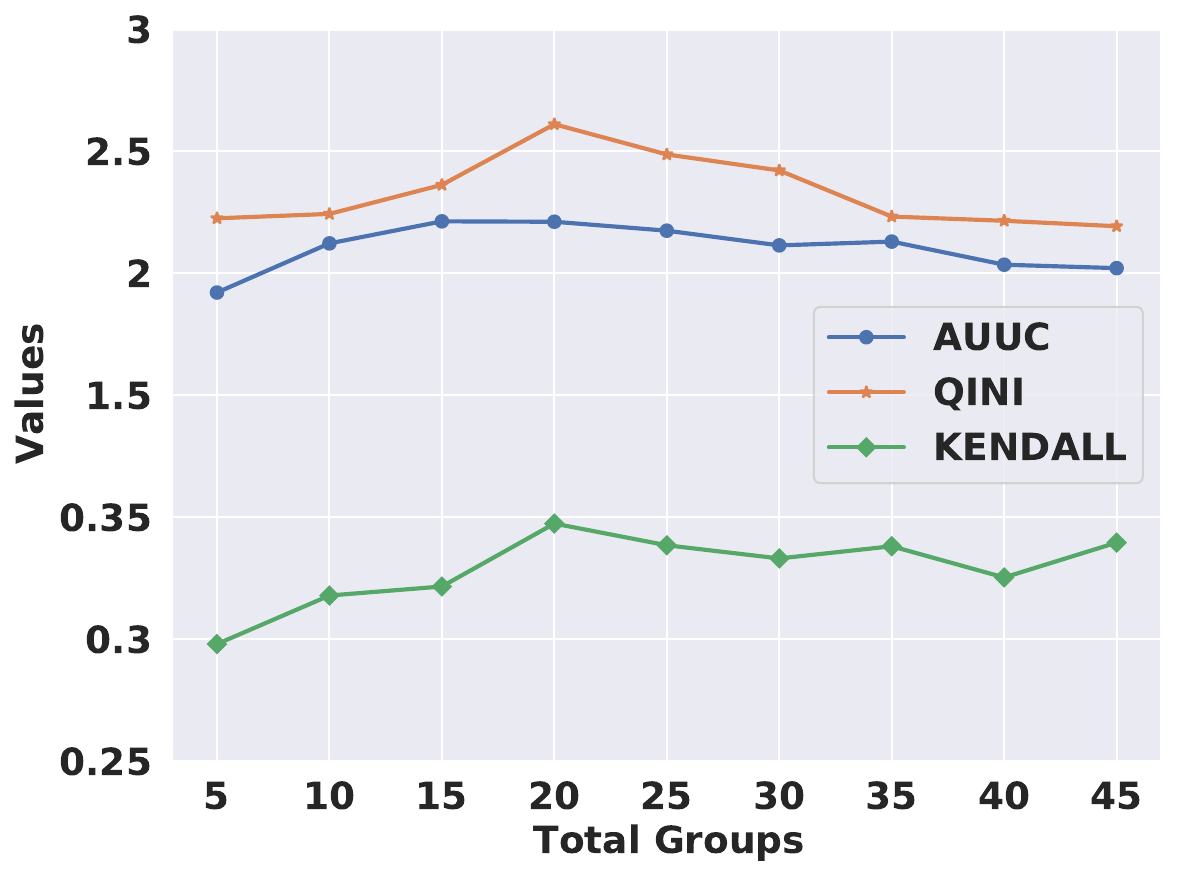}}
	\caption{The evaluation of different group number $K$ (5-45). We report the mean over five runs with different seeds.}
	\label{fig:pro_clu_performance}
\end{figure}
\begin{figure*}[!t]
	\centering  \includegraphics[width=0.8\linewidth]{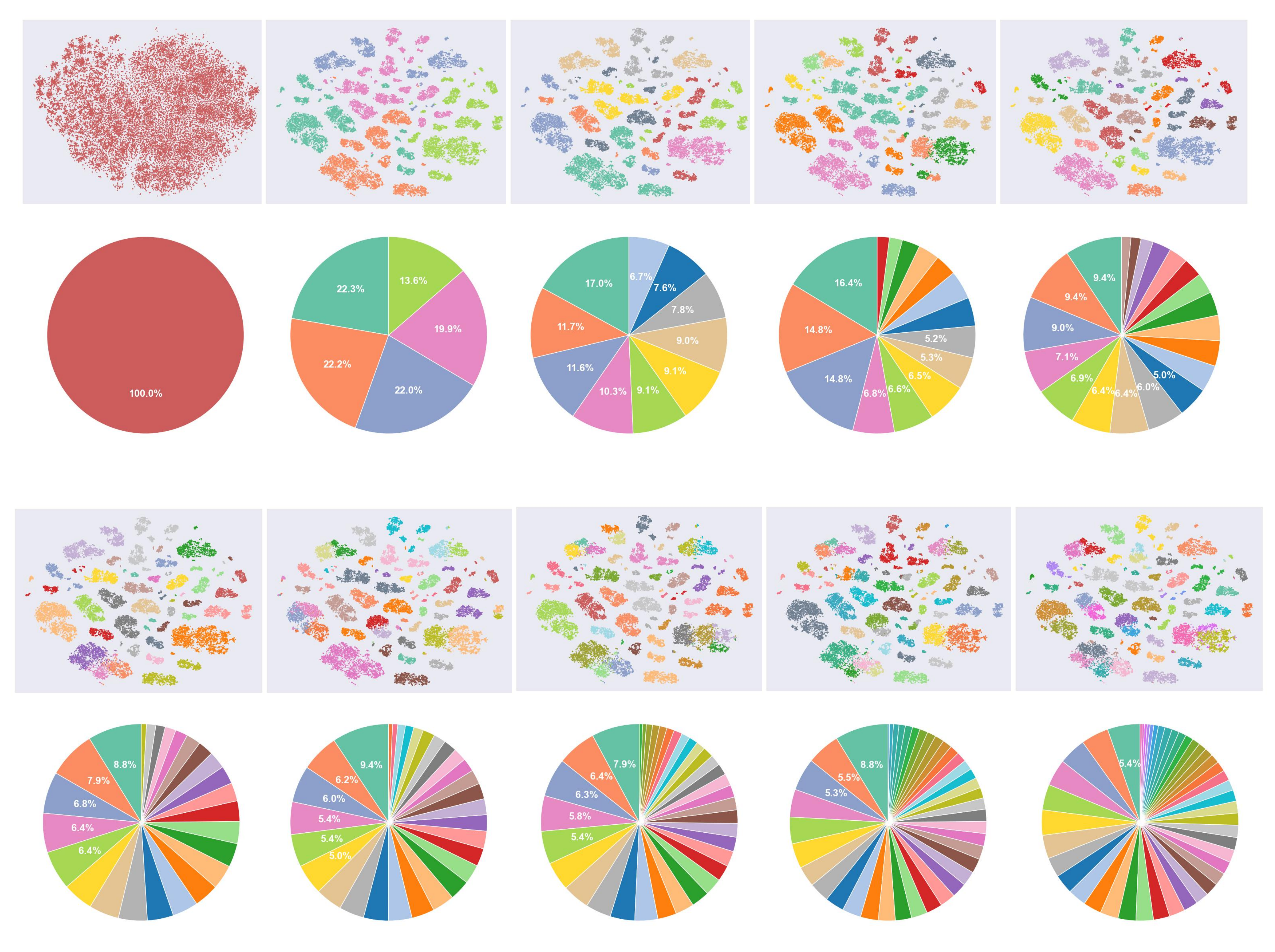}
	\caption{The context embedding t-SNE visualization of different group number $K$ (5-45, 5 per step) and the sample counts on the trained context embedding of Production dataset. The first is the distribution of the normalized original context data.}
	\label{fig:clu_pro}
\end{figure*}
\subsection{Hyperparameter Sensitivity analysis}

In this section, we conduct experiments on the two datasets to evaluate the influence of the important hyperparameters, \textit{i.e.}, batch size $b$, loss weight $\beta$, and loss weight $\gamma$. 
%==
In particular, we use the UMLC (CFRNet-mmd) as the model to handle the experiments; we present the results in Figure~\ref{fig:hyper}.

From the results, we have the following observations. 1) Our UMLC is pretty robust to the batch size $b$, even though we adjust the sample weights on the batch level. This shows that the information gain we used by the treatment feature interaction is a soft approach for the model optimization, based on the trained user-context embedding, the loss weight may be rational for the model optimization, which does not influence the model performance much.
%==
2) For the loss weight $\beta$ and $\gamma$, the model performance first increases across various metrics, but it suffers from the decrease when the weights are bigger than the specific values.
%==
This may be because the two components of the loss function can be regarded as the regularizations of the model; the big weights may cause the model to lose the power of prediction. 
%==
Thus, the readers need to carefully adjust their values on the different datasets.

\begin{table*}[!t]
    \centering
    \caption{Additional Ablation study of our UMLC with another three uplift models on Synthetic and Production datasets. We report the results over five random seeds. The best results and second best results are \textbf{bold} and \underline{underlined}, respectively.}    \label{tab:add_ablation}
    \begin{tabular}{c|ccccccc}\toprule
\multirow{2}{*}{\textbf{Method}} & \multicolumn{3}{c}{Synthetic Dataset} & &\multicolumn{3}{c}{Production Dataset}  \\\cline{2-4} \cline{6-8}
 & AUUC & QINI & KENDALL & & AUUC & QINI & KENDALL\\
 \midrule
\rowcolor{mygray} \textbf{UMLC} (S-Learner) & \textbf{0.2346} $\pm$ 0.0226 & \textbf{0.2597} $\pm$ 0.0283 & \textbf{0.1304} $\pm$ 0.0126 & & \textbf{1.9038} $\pm$ 0.0583 & \textbf{2.4983} $\pm$ 0.0518 & \textbf{0.2792} $\pm$ 0.0237 \\ \midrule
 w/o RCG & 0.2145 $\pm$ 0.0264 & 0.2512 $\pm$ 0.0233 & 0.1207 $\pm$ 0.0135  & & 1.8263 $\pm$ 0.0536 & 2.3252 $\pm$ 0.0483 & 0.2613 $\pm$ 0.0236  \\
 w/o UCI & 0.2259 $\pm$ 0.0225 & \underline{0.2519} $\pm$ 0.0214 & 0.1240 $\pm$ 0.0137  & & 1.8342 $\pm$ 0.0547 & \underline{2.4552} $\pm$ 0.0462 & 0.2641 $\pm$ 0.0263 \\
 w/o TFI & \underline{0.2343} $\pm$ 0.0204 & 0.2585 $\pm$ 0.0195 & \underline{0.1296} $\pm$ 0.0147  & & \underline{1.8591} $\pm$ 0.0556 & 2.4205 $\pm$ 0.0483 & \underline{0.2763} $\pm$ 0.0247 \\ \midrule
\rowcolor{mygray}  \textbf{UMLC} (T-Learner) & \textbf{0.2734} $\pm$ 0.0284 & \textbf{0.2783} $\pm$ 0.0223 &\textbf{0.1279} $\pm$ 0.0148 & & \textbf{2.0321} $\pm$ 0.0525 &\textbf{2.4983} $\pm$ 0.0527 & \textbf{0.3174} $\pm$ 0.0237\\ \midrule
 w/o RCG & 0.2291 $\pm$ 0.0282 & 0.2230 $\pm$ 0.0226 & 0.1143 $\pm$ 0.0184 & & 1.8133 $\pm$ 0.0522 & 2.3531 $\pm$ 0.0438 & 0.2729 $\pm$ 0.0274 \\
 w/o UCI & 0.2465 $\pm$ 0.0235 & 0.2483 $\pm$ 0.0194 & 0.1171 $\pm$ 0.0136 & & 1.8781 $\pm$ 0.0573 & 2.1255 $\pm$ 0.0462 & \underline{0.2982} $\pm$ 0.0224 \\
 w/o TFI & \underline{0.2657}  $\pm$ 0.0228 & \underline{0.2598} $\pm$ 0.0122 & \underline{0.1215} $\pm$ 0.0137 & & \underline{2.0201}$\pm$ 0.0583 & \underline{2.4218} $\pm$ 0.0542 & 0.2947 $\pm$ 0.0295 \\ \midrule
\rowcolor{mygray}  \textbf{UMLC} (TARNet) & \textbf{0.2913} $\pm$ 0.0294 &\textbf{0.2472} $\pm$ 0.0228 & \textbf{0.1219} $\pm$ 0.0162 & &\textbf{2.1349} $\pm$ 0.0552 & \underline{2.3842} $\pm$ 0.0454 & \textbf{0.3306} $\pm$ 0.0273 \\\midrule
 w/o RCG & 0.2296 $\pm$ 0.0235 & 0.2230 $\pm$ 0.0223 & \underline{0.1194} $\pm$ 0.0146 & & 1.9073 $\pm$ 0.0584 & \textbf{2.3983} $\pm$ 0.0525 & 0.3278 $\pm$ 0.0235 \\
 w/o UCI & 0.2567 $\pm$ 0.0292 & 0.2268 $\pm$ 0.0235 & 0.1187 $\pm$ 0.0184 & & 1.8578 $\pm$ 0.0587 & 2.3255 $\pm$ 0.0484 & \underline{0.3284} $\pm$ 0.0271 \\
 w/o TFI & \underline{0.2783} $\pm$ 0.0223 & \underline{0.2445} $\pm$ 0.0209 & 0.1175 $\pm$ 0.0142 & & \underline{2.0422} $\pm$ 0.0554 & 2.3365 $\pm$ 0.0427 & 0.3263 $\pm$ 0.0235 \\ \bottomrule
    \end{tabular}
    \label{tab:my_label}
\end{table*}

\subsection{Context Grouping Analysis}

To further validate the effectiveness of the response-guided context grouping module, we also conducted experiments to show the visualization, the alignment, and the performance of our UMLC on the Production dataset.
%==
The alignment and the performance of the learned model (\textit{i.e.}, UMLC (CFRNet-mmd)) are shown in Figure~\ref{fig:clu_pro}, the visualization is shown in Figure~\ref{fig:clu_pro}.
%==
Due to the complexity of the context data in the Production dataset, the true group number of the learned context embedding is unknown, we show the cluster visualization and sample counter from 5 to 45 (5 per step).
%==
To reflect the context grouping performance more clearly, instead of bar charts,  we use pie charts for the sample counter of each group number $K$ on the Production dataset.
%==
The sample counter, which is bigger than 5\% of the total samples, is indicated; with bigger $K$, the sample counters of most groups are less than 5\%. 
We combine the results in Figure~\ref{fig:pro_clu_performance} with the visualization to further analyze the performance.
%==
From Figure~\ref{fig:pro_clu_performance}(a), the context groups in the treatment and control groups become imbalanced with big total groups $K$ (i.e., 25, 35, 40, 45).
%==
This may be because the Production dataset is collected from a real-world scenario, which is more complex than the Synthetic dataset. 
%==
However, considering the alignment results and the model performance, they suffer from a big decline after the group number $K$ is bigger than 20.
%==
This indicates that, even with more complex data, we can target the number of context groups without training the whole model.

\subsection{Additional Ablation Study}
In the main contents, we conduct the experiments with four base uplift models. Due to the space limitation of the main contents, we present the experimental results with three another uplift models in Table~\ref{tab:add_ablation}.
%==
From the results, we can see that, with our UMLC, we can consistently improve the model performance on most metrics. 
%==
And the ablation study shows that removing each part of our UMLC will hurt the model performance.

\subsection{Additional Evaluations}\label{app:add_eval}
In the main contents of our paper, we focus on comparing each base model with and without UMLC. 
%==
To better evaluate the performance of our UMLC, in this section, we conduct the variant by three aspects, (1) without context grouping, (2) grouped with Kmeans, and (3) grouped with other cluster methods.
%==
The results are presented in Figure~\ref{fig:additional_eval}.
%==
From the results, we can see that our UMLC performs better than only using user features, using the category feature for grouping or applying Kmeans on the context samples. 
%==
For the variants of UMLC (\textit{i.e}., UMLC-Hierarchical, UMLC-Spectral and UMLC-DBSCAN), they perform similar with the original UMLC, which verifies the robustness of our UMLC for different cluster methods.

\begin{figure*}[!t]
    \centering
    \subfigure[Synthetic dataset]{\includegraphics[width=1\linewidth]{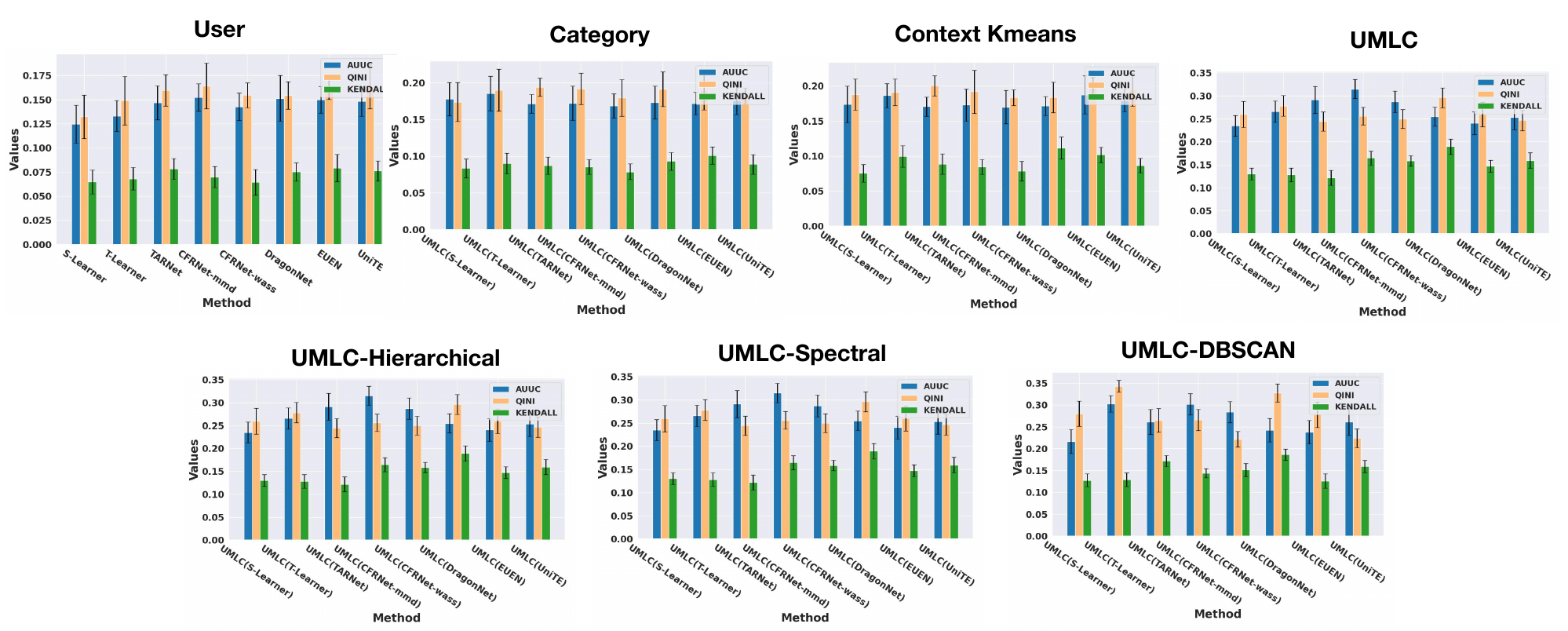}}\\
    \subfigure[Production dataset]{\includegraphics[width=1\linewidth]{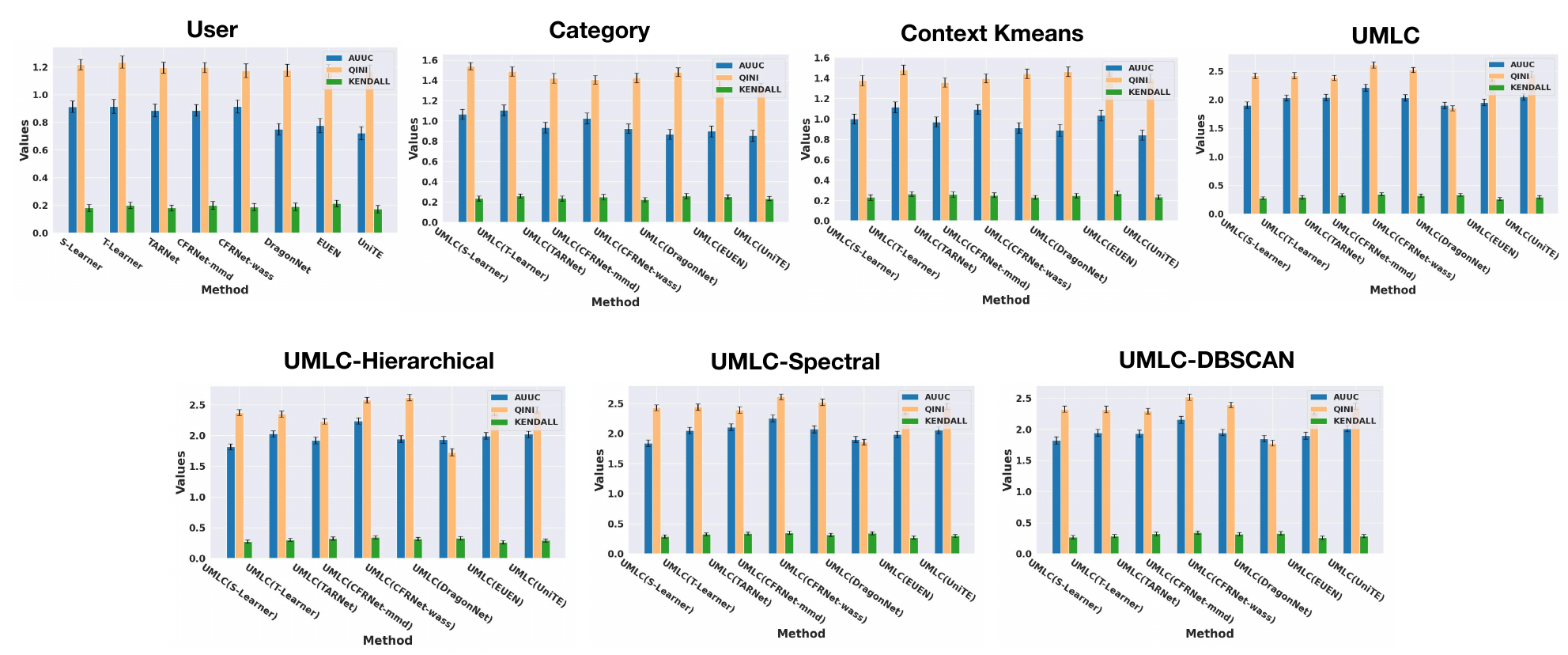}}
    \caption{Performance of the variations of our UMLC. User represents the results with only the user features. Category represents we use the ground truth category in the context features for learning the context embedding. Context K-means represents that we apply K-means on the context samples rather than embedding. UMLC-Hierarchical, UMLC-Spectral and UMLC-DBSCAN are the variants of our UMLC with different cluster methods on the context embedding.}
    \label{fig:additional_eval}
\end{figure*}

\end{document}